\documentclass[11pt]{article}
\usepackage{amstext, amsmath,latexsym,amsbsy,amssymb,amsmath}

\usepackage{enumitem}

\newtheorem{proposition}{Proposition}[section]

\newenvironment{proof}{{\bf Proof\ }}{\QED\\}
\newtheorem{lemma}{Lemma}[section]

\numberwithin{equation}{section}
\newtheorem{theorem}{Theorem}[section]

\newcommand{\QED}{\hspace*{\fill}\rule{2.5mm}{2.5mm}}

\usepackage{amssymb}\usepackage{graphicx}
\usepackage{tikz}

\usepackage{color}
\newcommand\qed{\hfill$\sqcap\kern-7.5pt\hbox{$\sqcup$}$}

\newcommand{\beqn}{\begin{equation}}
\newcommand{\eeqn}{\end{equation}}
\newcommand{\bear}{\begin{eqnarray}}
\newcommand{\eear}{\end{eqnarray}}
\newcommand{\bean}{\begin{eqnarray*}}
\newcommand{\eean}{\end{eqnarray*}}

\allowdisplaybreaks

\begin{document}
\title{On coupling kinetic and Schr\"odinger equations}

\author{Avy Soffer\footnotemark[1] \and Minh-Binh Tran\footnotemark[2] 
}

\renewcommand{\thefootnote}{\fnsymbol{footnote}}

\footnotetext[1]{Mathematics Department, Rutgers University, New Brunswick, NJ 08903 USA.\\Email: soffer@math.rutgers.edu
}

\footnotetext[2]{Department of Mathematics, University of Wisconsin-Madison, Madison, WI 53706, USA. \\Email: mtran23@wisc.edu
}

\maketitle
\begin{abstract} 
We consider in this paper a system coupling a linear quantum Boltzmann equation and a defocusing cubic nonlinear Schr\"odinger equation. The Schrodinger equation reflects the dynamics of the wave function of the Bose-Einstein Condensate and the kinetic part of the system describes the evolution of the density function of the thermal cloud. An existence and uniqueness result for the system is supplied. We also prove the convergence to equilibrium of the density function of the thermal cloud and a scattering theory for the wave function of the condensate.
 \end{abstract}

{\bf Keyword:}
{Low and high temperature quantum kinetics; Bose-Einstein  condensate; quantum Boltzmann equation;  defocusing cubic nonlinear Schrodinger equation; scattering theory; convergence to equilibrium. 


\tableofcontents
\section{Introduction}
When a bose gas  is cooled below the Bose-Einstein critical temperature, the
Bose-Einstein condensate is formed, consisting of a macroscopic number of particles, all in the ground state of the system. A finite temperature trapped Bose gas is composed of two distinct components, the Bose-Einstein Condensate and the noncondensate - thermal cloud.  Since the initial discoveries of Bose-Einstein Condensates (BECs) by the JILA and MIT groups, there has been much experimental and theoretical research on BECs and their thermal clouds (see \cite{semikoz1997condensation,kagan1997evolution,KD1,KD2,ZarembaNikuniGriffin:1999:DOT,proukakis2008finite,Spohn:2010:KOT,QK0,QK1,QK2,ReichlGust:2013:TTF,ReichlGust:2013:RRA,ReichlGust:2012:CII}, and references therein). The first model for the system of the interaction between BECs and their thermal clouds was introduced by Kirkpatrick and Dorfman in \cite{KD1,KD2}. By a simpler technique, the model was  revisited by Zaremba, Nikuni, Griffin in \cite{ZarembaNikuniGriffin:1999:DOT}. The terminology ``Quantum Kinetic Theory'' was first introduced by Gardinier, Zoller and collaborators in the series of papers \cite{QK0,QK1,QK2}. Gardinier and Zoller's Master Quantum Kinetic Equation, at the limit, returns to  the Kirkpatrick-Dorfman-Zaremba-Nikuni-Griffin (KDZNG) model. For more discussions and references on this topic, we refer to the review paper \cite{anglin2002bose} and the books \cite{inguscio1999bose,ColdAtoms1}.

Let us note that besides the kinetic theory point of view, there are other approaches to the study of BECs and excitations: the excitation spectrum \cite{seiringer2011excitation}, Fock space approach used to improve convergence rate in the analysis of Hepp, Rodnianski-Schlein \cite{GrillakisMachedon:2013:BMF,GrillakisMachedonMargetis:2011:SOC,GrillakisMachedon:2013:PEA}, Fock space approach central limit theorem \cite{KirkpatrickSchlein:2013:ACL}, Quasifree reduction \cite{bach2016time}, the time evolution of the one-particle wave function of an excitation \cite{DeckertFrohlichPickl:2016:DSW,mitrouskas2016bogoliubov}, and cited references. Quantum kinetic theory, on the other hand, is both a genuine kinetic theory and a genuine quantum theory. In which, the kinetic part arises from the decorrelation between different momentum bands.

During the last 10 years, quantum kinetic theory has also become as an important topic with a lot of interest (see \cite{LevermoreLiuPego:2016:GDB,Pego:2012:SSO,ChenGuo:2015:OTW,BenedettoCastella:2004:SCO,BenedettoPulvirenti:2005:OTW,Erdos:2004:OTQ,Carrillo:2016:TFP,EscobedoBinh,EscobedoVelazquez:2015:FTB,GambaSmithBinh,JinBinh,AlonsoGambaBinh,ToanBinh,ReichlTran,CraciunBinh,germain2017optimal,SofferBinh1,nguyen2017quantum} and references therein). According to the theory, the density function of the thermal cloud satisfies a quantum Boltzmann equation and the wave function of the  condensate follows the nonlinear Gross-Pitaevskii equation. The coupled dynamics of  the kinetic and Gross-Pitaevskii equations brings in a whole new class of phenomena.

In this  work, we are interested in the long time dynamics of the kinetics-Schrodinger coupling system. In order to explain the physical intuition behind our work, let us look at the classical example of surface waves on ocean. Wind blowing along the air-water interface is  what creates ocean surface waves. As wind continues to blow, it forms a steady disturbance on the surface, that leads to the rise of  the wave crests. Surface waves are the waves we see at  beaches and they occur  all over the globe. The coupling of the two states of matter gas - liquid  of this phenomenon is, in some sense, very similar to the coupling thermal cloud - Bose-Einstein Condensate. Suppose that the air above the ocean, after quite a long time, stands still and its density function reaches the equilibrium distribution. If this happens, we will not see ocean waves anymore, under the assumption that tidal waves,  tsunamis and other waves are negligible. That means the ocean wave function   is scattered into a constant function. This gives us an intuition for the long time dynamics of the system thermal cloud - Bose-Einstein Condensate: the thermal cloud would converge to equilibrium, as a normal gas; on the other hand, the wave function of the condensate would also converge to a constant function. In other words, we hope to prove  the convergence to equilibrium to the solution of the quantum Boltzmann equation, and the scattering theory for the solution of the nonlinear Schr\"odinger equation. 

With our current technologies, in order to study the scattering theory for the solution of the nonlinear Schr\"odinger equation, we will need to solve at least two problems:
\begin{itemize}
\item The nonhomogeneous quantum Boltzmann equation has a strong, unique global solution, since we will need to put this solution back into the Sch\"odinger equation to do the coupling.
\item The solution of the nonhomogeneous quantum Boltzmann equation converges to equilibrium with a sufficiently fast rate (exponentially), since the coupling behaves like a confining potential for the Sch\"odinger equation and we will want this potential to nicely behave. 
\end{itemize}

Unfortunately, both of these two problems  still remains opened, even in the context of the classical Boltzmann equation. Therefore, as the  first step to understand the long time dynamics of the kinetics-Sch\"odinger coupling, let us simplify the system by replacing the nonhomogeneous quantum Boltzmann equation by a linear quantum Boltzmann equation, and study the following coupling system, where $f(t,r,p)$ denotes the density function of the excitations at time $t$, position $r$ and momentum $p$ and $\Psi(t,r)$ is the wave function of the condensate at time $t$ and position $r$:
\begin{eqnarray}\label{System1}
\frac{\partial f}{\partial t}(t,r,p)+{p}\cdot\nabla_{{r}} f(t,r,p) & = &\ L[f](t,r,p), \\\nonumber
&&\ \ \ \ \ \ \ \ \ \ \ \ (t,r,p) \ \in \ \mathbb{R}_+\times\mathbb{R}^3\times\mathbb{R}^3,\\\label{System1b}
f(0,r,p)\ &=& \ f_0(r,p), (r,p)\in\mathbb{R}^3\times\mathbb{R}^3,\\\label{System2}
i  \frac{\partial \Psi(t,r)}{\partial t} \ &=& \ \Big(-{ \Delta_{{r}}} \ + \ |\Psi(t,r)|^2 + W(t,r) \Big)\Psi(t,r),\\
\Psi(0,r)&=&\Psi_0(r), \forall r\in\mathbb{R}^3,\\\label{System3}
\rho[f](t,r)&=&\int_{\mathbb{R}^3}f(t,r,p')dp',\\\label{System4}
\ N_c(t,r)&=&C^*\int_{\mathbb{R}^3}|\Psi|^2(t,r)e^{-|r-r'|}dr', \\\label{System5}
\ W(t,r)&=&-1-V(t,r), \ V(t,r) = - M_{f_0}+\rho(t,r),
\end{eqnarray}
where $L$ is of the form \eqref{Operator1} or \eqref{Operator2} and $\vartheta$ is some positive constant, $C^*_\vartheta$ is the normalization constant such that
$$C^*_\vartheta \int_{\mathbb{R}^3}e^{-|p|/\vartheta}dp=1.$$ 
For the sake of simplicity, let us set $\vartheta=1$ and denote $C^*_1$ as $C^*$.\\
 Moreover,
$$M_{f_0} \ = \ \int_{\mathbb{R}^3\times\mathbb{R}^3}f_0(r,p)drdp.$$
The Bose-Einstein distribution function is defined
\begin{equation}\label{BEDistribution}
\mathfrak{E}(p)=C_E \frac{1}{e^{\beta |p|}-1},
\end{equation} 
with $\beta :=\frac {1} {k_B T}>0$ is a given physical constant depending on the Boltzmann constant $k_B$, and the temperature of the quasiparticles $T$ at equilibrium. For the sake of simplicity, we suppose $\beta=1$.
The normalized constant $C_E$ is chosen such that
$$\int_{\mathbb{R}^3}\mathfrak{E}(p)dp=1.$$ 
We impose the following boundary condition on $\Psi$
\begin{equation}\label{GPBoundary}
\lim_{|r|\to\infty}\Psi=1.
\end{equation}
For more physical  background of the boundary condition \eqref{GPBoundary}, we refer to \cite{FetterSvidzinsky:2001:VIA,JonesRoberts:1982:MIA,JonesPuttermanRoberts:1986:MIA,BethuelSaut:VAS:2002,SulemSulem:TNS:1999} and references therein. }
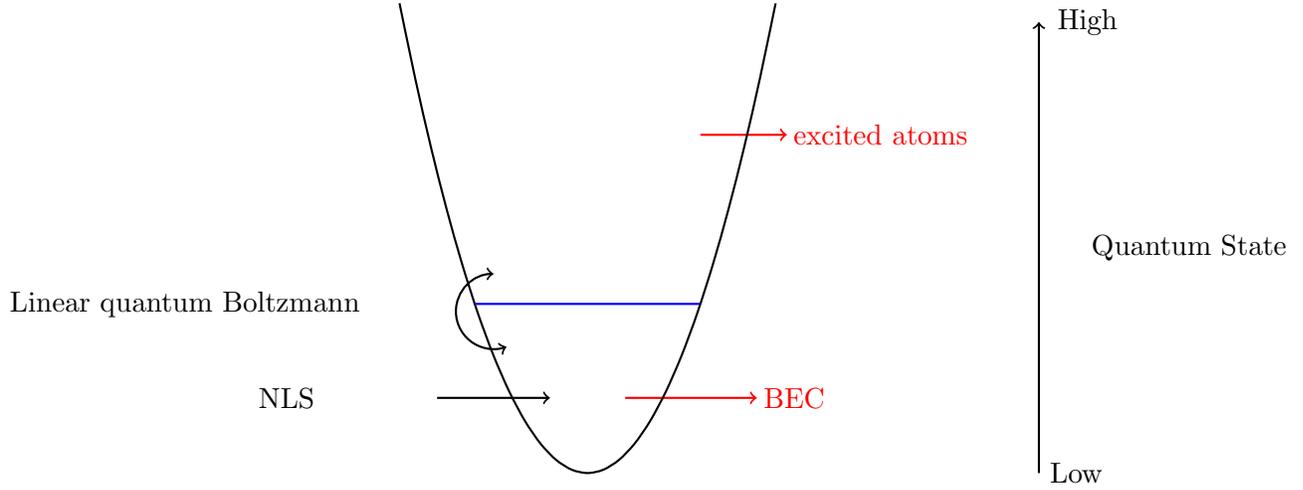
\begin{figure}

\begin{tikzpicture}
      \draw[thick, ->, red] (1.5,4.5) -- ++(1.15,0) node[xshift=1.25cm] {excited atoms};
      \draw[thick, ->] (-2,1) -- ++(1.5,0) node[xshift=-3.5cm] {NLS};
      \draw[thick, ->, red] (0.5,1) -- ++(1.75,0) node[xshift=.5cm] {BEC};
      \draw[thick, ->] (6,0) node [xshift=.5cm]{Low} -- ++(0,3)  node[xshift=2.0cm] {Quantum State} -- ++(0,3) node[xshift=.65cm] {High};
      \draw[ thick,domain=-2.5:2.5,smooth,variable=\x,black] plot ({\x},{\x*\x});
      \draw [thick, blue] (-1.5, 2.25) -- (1.5, 2.25);

      \draw[thick, <->] (-1.25,2.65) arc (90:290:.5);
      \node at (-5.35, 2.25) {Linear quantum Boltzmann};
    \end{tikzpicture}
    \label{fig2}\caption{The simplified model of the Bose-Einstein Condensate (BEC) and the excited atoms.}

\end{figure}
\\ We define 
$$\mathcal{L}:=\left\{f \ \Big| \ \|f\|_{\mathcal{L}}:=\left(\int_{\mathbb{R}^3\times\mathbb{R}^3}|f|^2\mathfrak{E}^{-1}drdp\right)^{1/2} \ < \ \infty\right\},$$
 denote the Lebesgue and the Sobolev 
spaces by $L^p$, $H^{s,p}$ respectively, for $1 \le p, q \le \infty$ and $s\in\mathbb{R}$.

Note that the  nonlinear model has already been considered in our previous works \cite{SofferBinh1,AlonsoGambaBinh}, in which the kinetic and Schrodinger equations are decoupled.

 Define 
\begin{equation}\label{KineticEquilibrium}
f_\infty (p) =M_{f_0} \mathfrak{E}(p),
\end{equation}
 The main Theorem of our paper is the following:
\begin{theorem}\label{Thm:Main}
Suppose that $f_0$ be a positive function in $L^1(\mathbb{R}^3\times\mathbb{R}^3)\cap \mathcal{L}$. There exists $\delta>0$ such that for $\|\nabla N_c(0,\cdot) \|_{L^\infty_r} \le \delta$ and if  
$\Psi_0\in H^1(\mathbb{R}^3)$ satisfies
\begin{equation}\label{Thm:Main:1}
\int_{\mathbb{R}^3}\langle r\rangle^2 \left(|\mathrm{Re}\Psi_0(r)|^2+|\nabla \Psi_0(r)|^2\right)dr<\delta^2,
\end{equation}
 under the assumption that $L$ is of the form \eqref{Operator2}, the System \eqref{System1}-\eqref{BEDistribution} has a unique solution $(f,\Psi)$. The first component $f\in C^1(\mathbb{R}_+,\mathcal{L})$, $f\geq 0$ and $f$ decays exponentially in time towards the equilibrium $f_\infty$  in the following sense: there exist  $\mathcal{C}_1,\mathcal{C}_2>0$ depending only on $\mathfrak{E}$, $\delta$, such that
\begin{equation}
\label{Thm:Main:2}
\|f(t)-f_\infty\|_{\mathcal{L}} \ + \ \|\nabla f(t)\|_{\mathcal{L}}\le \mathcal{C}_1e^{-\mathcal{C}_2t}. 
\end{equation}
Moreover, there also exists $\mathcal{C}_3>0$ depending only on $\mathfrak{E},\delta$ such that
\begin{equation}
\label{Thm:Main:3}
\|\rho[f](t)-M_{f_0}\|_{L^2_r(\mathbb{R}^3)}\le {\|f_0-f_\infty\|_{\mathcal{L}}}e^{-\mathcal{C}_3 t}.
\end{equation}
Define $U$ and $H$ as in \eqref{Sec:NLS:E6} and \eqref{Sec:NLS:E8}. The second component satisfies $\Psi=1+u$ and  for $v:=\mathrm{Re} u +iU\mathrm{Im}u$, we have $e^{itH}v\in C(\mathbb{R}; \langle r\rangle^{-1}H^1_r(\mathbb{R}^3))$. Moreover,
\begin{equation}\label{Thm:Main:4}
\|v(t)-e^{-itH}v_+\|_{H^1_r} \ \le O\left({{(t+1)}^{-1/2}}\right), \  \|\langle r\rangle \left[e^{itH}v(t)-v_+\right]\|_{H^1_r} \to 0,
\end{equation}
as $t\to 0$, for some $v_+\in \langle r\rangle^{-1} H^1_r(\mathbb{R}^3)$.\\
Define $u=u_{1}+ iu_{2}$, we also have
\begin{equation}\label{Thm:Main:5}
\|u_{1}(t)\|_{L^\infty}\leq O\left({{(t+1)}^{-1}}\right),~~~~\|u_{2}(t)\|_{L^\infty}\leq O\left({{(t+1)}^{-9/10}}\right).\end{equation}
\end{theorem}
The proof of the theorem requires mixed techniques coming from both topics: kinetic
and Schr\"odinger equations. To be more precise, we propose a new framework, in which a
convergence to equilibrium technique is combined with a normal form transformation, to
study the long time asymptotics of the system.

Notice that in the above theorem, we choose $L$ to be of the form \eqref{Operator2}. We will see in Proposition \ref{PropNoE:Kinetic} that in the case $L$ is of the form \eqref{Operator1} we get a polynomial  decay in time of the convergence to equilibrium. This decay rate is to weak for the scattering theory of the Schrodinger equation to be true. On the other hand, in Proposition \ref{Pro:Kinetic}, when $L$ is of the form \eqref{Operator2}, the convergence rate to equilibrium is exponential in time.

The structure of the paper is as follows:  Section $2$ is devoted to the explication of how to obtain \eqref{System1}-\eqref{System5} from the quantum kinetic - Schrodinger system describing the dynamics of a BEC and its thermal cloud (cf. \cite{KD1,KD2,ZarembaNikuniGriffin:1999:DOT,GriffinNikuniZaremba:2009:BCG,QK0,QK1,QK2,QK3,QK4,QK5,QK6,ReichlGust:2013:TTF,ReichlGust:2013:RRA,ReichlGust:2012:CII}). In Propositions \ref{PropNoE:Kinetic} and  \ref{Pro:Kinetic}, we provide  the existence, uniqueness and convergence to equilibrium results of the linear quantum Boltzmann equation, for two different choices of the operator $L$: \eqref{Operator1} and \eqref{Operator2}. Proposition \ref{Pro:NLS} discusses existence and uniqueness results  for the nonlinear Schrodinger equation as well as the scattering theory for the equation. Based on Propositions   \ref{Pro:Kinetic} and \ref{Pro:NLS}, the proof of Theorem \ref{Thm:Main} is supplied in Section 4.
\section{The simplified model on the coupling between Schrodinger and kinetic equations}
In this section, we explain how to obtain the System \eqref{System1}-\eqref{System5} from the  quantum kinetic - Schrodinger system describing the dynamics of a BEC and its thermal cloud (cf. \cite{KD1,KD2,ZarembaNikuniGriffin:1999:DOT,GriffinNikuniZaremba:2009:BCG,QK0,QK1,QK2,QK3,QK4,QK5,QK6,ReichlGust:2013:TTF,ReichlGust:2013:RRA,ReichlGust:2012:CII}).
First, recall the BEC-thermal cloud system, at moderately low temperature regime:
\begin{eqnarray}\label{QBFull}
\frac{\partial f}{\partial t}(t,r,p)&+&{p}\cdot\nabla_{{r}} f(t,r,p)\\\nonumber
&=&Q[f](t,r,p):=n_c(t,r)C_{12}[f](t,r,p)+C_{22}[f](t,r,p), (t,r,p)\in\mathbb{R}_+\times\mathbb{R}^3\times\mathbb{R}^3,\\\nonumber
f(0,r,p)&=&f_0(r,p), (r,p)\in\mathbb{R}^3\times\mathbb{R}^3,\\\nonumber
C_{12}[f](t,r,p_1)&:=&{\frac{2g^2}{(2\pi)^2\hbar^4}}\iint_{\mathbb{R}^{3}\times\mathbb{R}^{3}}\delta({p}_1-{p}_2-{p}_3)\delta(\mathcal{E}_{{p}_1}-\mathcal{E}_{{p}_2}-\mathcal{E}_{{p}_3})\\\label{C12}
& &\times[(1+f(t,r,{p}_1))f(t,r,{p}_2)f(t,r,{p}_3)-\\\nonumber
&&-f(t,r,{p}_1)(1+f(t,r,{p}_2))(1+f(t,r,{p}_3))]d{p}_2d{p}_3\\\nonumber
&&-2{{\frac{2g^2}{(2\pi)^2\hbar^4}}}\iint_{\mathbb{R}^{3}\times\mathbb{R}^{3}}\delta({p}_2-{p}_1-{p}_3)\delta(\mathcal{E}_{{p}_2}-\mathcal{E}_{{p}_1}-\mathcal{E}_{{p}_3})\\\nonumber
& &\times[(1+f(t,r,{p}_2))f(t,r,{p}_1)f(t,r,{p}_3)-\\\nonumber
&&-f(t,r,{p}_2)(1+f(t,r,{p}_1))(1+f(t,r,{p}_3))]d{p}_2d{p}_3,\\\label{C22}
C_{22}[f](t,r,p_1)&:=&\frac{2g^2}{(2\pi)^5\hbar^7}\iiint_{\mathbb{R}^{3}\times\mathbb{R}^{3}\times\mathbb{R}^{3}}\delta({p}_1+{p}_2-{p}_3-{p}_4)\\\nonumber
& &\times\delta(\mathcal{E}_{{p}_1}+\mathcal{E}_{{p}_2}-\mathcal{E}_{{p}_3}-\mathcal{E}_{{p}_4})\times\\\nonumber
&&\times [(1+f(t,r,{p}_1))(1+f(t,r,{p}_2))f(t,r,{p}_3)f(t,r,{p}_4)\\\nonumber
&&-f(t,r,{p}_1)f(t,r,{p}_2)(1+f(t,r,{p}_3))(1+f(t,r,{p}_4))]d{p}_2d{p}_3d{p}_4,
\end{eqnarray}
where  $n_c(t,r)=|\Phi|^2(t,r)$ is the condensate density, $\Phi$ satisfies 
\begin{equation}
\begin{aligned}\label{GP}
i \hbar \frac{\partial \Phi(t,r)}{\partial t}=&\ \Big(-\frac{\hbar \Delta_{{r}}}{2m}+g|\Phi(t,r)|^2+2g\int_{\mathbb{R}^3}fdp+\frac{ig^2}{2\hbar}\iiint_{\mathbb{R}^{3}\times\mathbb{R}^{3}\times\mathbb{R}^{3}}\delta({p}_1-{p}_2-{p}_3)\\
&\times\delta(\mathcal{E}_{{p}_1}-\mathcal{E}_{{p}_2}-\mathcal{E}_{{p}_3})[(1+f(t,r,{p}_1))f(t,r,{p}_2)f(t,r,{p}_3)-\\
&-f(t,r,{p}_1)(1+f(t,r,{p}_2))(1+f(t,r,{p}_3))]dp_1d{p}_2d{p}_3\Big)\Phi(t,r), \ \ (t,r)\in\mathbb{R}_+\times\mathbb{R}^3,\\
~~~\Phi(0,r)=& \ \Phi_0(r), \forall r\in\mathbb{R}^3,
\end{aligned}
\end{equation}

 \begin{figure}

\begin{tikzpicture}
      \draw[thick, ->] (-2.5,4.5) -- ++(1,0) node[xshift=-1.5cm] {$C_{22}$};
      \draw[thick, ->, red] (1.5,4.5) -- ++(1.15,0) node[xshift=1.25cm] {excited atoms};
      \draw[thick, ->] (-2,1) -- ++(1.5,0) node[xshift=-2.25cm] {NLS};
      \draw[thick, ->, red] (0.5,1) -- ++(1.75,0) node[xshift=.5cm] {BEC};
      \draw[thick, ->] (6,0) node [xshift=.5cm]{Low} -- ++(0,3)  node[xshift=2.0cm] {Quantum State} -- ++(0,3) node[xshift=.65cm] {High};
      \draw[ thick,domain=-2.5:2.5,smooth,variable=\x,black] plot ({\x},{\x*\x});
      \draw [thick, blue] (-1.5, 2.25) -- (1.5, 2.25);

      \draw[thick, <->] (-1.25,2.65) arc (90:290:.5);
      \node at (-2.15, 2.25) {$C_{12}$};
    \end{tikzpicture}
    \label{fig}\caption{The Bose-Einstein Condensate (BEC) and the excited atoms.}
    
   \end{figure}
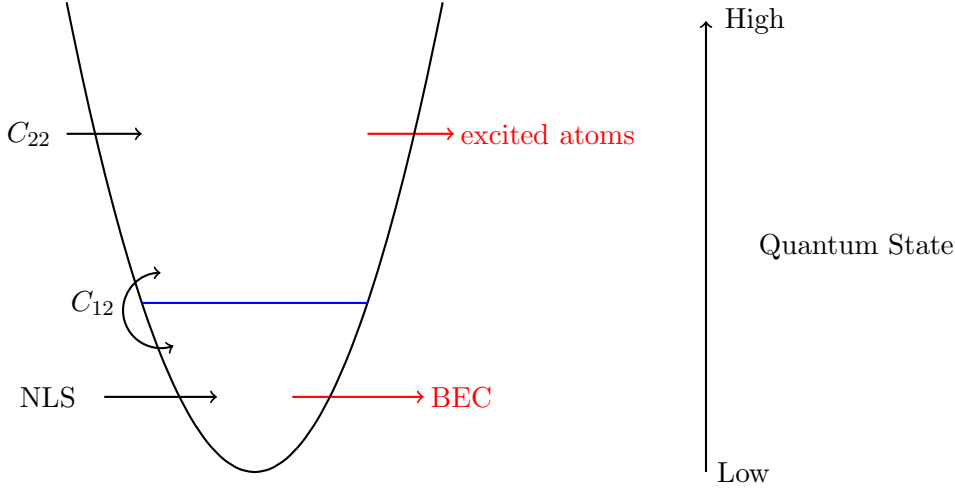
 and $\mathcal{E}_{{p}}$ is the Bogoliubov dispersion law
\bear \label{def-E}
\mathcal{E}_p\ =\ \mathcal{E}(p)\ =\ \sqrt{\kappa_1 |p|^2 + \kappa_2 |p|^4}, \qquad \kappa_1\ =\ \frac{gn_c}{m}>0, \quad \kappa_2\ =\ \frac{1}{4m^2}>0,
\eear
 $m$ is the mass of the particles,  $g$ is the interaction coupling constant.
\\ Notice that \eqref{C12} describes collisions of the condensate and the non-condensate atoms (condensate growth term), \eqref{C22} describes collisions between non-condensate atoms, and \eqref{GP} is the defocusing nonlinear Schr\"odinger equation of the condensate. \\
We assume that the temperature of the system is low enough, such that  collisions of the condensate and the non-condensate atoms are much stronger than the collisions between non-condensate atoms,  $C_{22}$ is therefore negligible. The BEC-thermal cloud system is reduced to
\begin{equation}\label{QBC12}
\begin{aligned}
\frac{\partial f}{\partial t}(t,r,p)\ + & \ {p}\cdot\nabla_{{r}} f(t,r,p) \\
= & \ Q[f](t,r,p):=n_c(t,r)C_{12}[f](t,r,p), (t,r,p)\in\mathbb{R}_+\times\mathbb{R}^3\times\mathbb{R}^3,\\
f(0,r,p)\ =& \ f_0(r,p), (r,p)\in\mathbb{R}^3\times\mathbb{R}^3,
\end{aligned}
\end{equation}
where  $n_c=|\Phi|^2$ is the condensate density, $\Phi$ satisfies \eqref{GP}.
\\
Notice that $\Phi(t,r)$ is usually a function in $H^1_r(\mathbb{R}^3))$, it is not easy to evaluate the value of it at each point $(t,r)$. We therefore replace $n_c$ by the average  $N_c$
$$N_c(t,r)=C^*\int_{\mathbb{R}^3}|\Phi|^2(t,r)e^{-|r-r'|^2 }dr',$$
where $\vartheta$ is some positive constant and $C^*$ is the normalized constant
$$C^*\int_{\mathbb{R}^3}e^{-|r'|^2 }dr'=1.$$
Since the coupling between the equation involving $C_{12}[f]$ and the cubic nonlinear Schr\"odinger equation is difficult to study, let us replace $C_ {12}[f]$ by the linear quantum Boltzmann operator $L[f]$:
\begin{eqnarray}\label{LB}
\frac{\partial f}{\partial t}(t,r,p) \ + \ {p}\cdot\nabla_{{r}} f(t,r,p) 
&=&N_c(t,r)L[f](t,r,p), \mbox{ on } \mathbb{R}_+\times\mathbb{R}^3\times\mathbb{R}^3,\\
f(0,r,p)&=&f_0(r,p), (r,p)\in\mathbb{R}^3\times\mathbb{R}^3,
\end{eqnarray}
where 
$L[f]$ could be either $L_1[f]$ or $L_2[f]$
\begin{equation}\label{Operator1}
L_1[f](r,p) \ =  \ \mathfrak{E}(p)\int_{\mathbb{R}^3}f(r,p')dp'-f(r,p)
\end{equation}
and
\begin{equation}\label{Operator2}
L_2[f](r,p) \ =  \ \mathfrak{E}^{-1}(p)\nabla_p\left(\mathfrak{E}^{-1}(p) \nabla_p (f-f_\infty)\right),
\end{equation}
and $\Phi$ satisfies \eqref{GP1}:
\begin{equation}
\begin{aligned}\label{GP1}
i \hbar \frac{\partial \Phi(t,r)}{\partial t}\ =&\ \Big(-\frac{\hbar \Delta_{{r}}}{2m}\ +\ g|\Phi(t,r)|^2\ +\ 2g\int_{\mathbb{R}^3}fdp\Big)\Phi(t,r),\\
~~~\Phi(0,r)\ =&\ \Phi_0(r), \forall r\in\mathbb{R}^3.
\end{aligned}
\end{equation}
In the above equation, we have dropped the term containing the collision operator

\begin{equation}
\begin{aligned}\label{GP2}
&\ \frac{ig^2}{2\hbar}\iiint_{\mathbb{R}^{3}\times\mathbb{R}^{3}\times\mathbb{R}^{3}}\delta({p}_1-{p}_2-{p}_3)\delta(\mathcal{E}_{{p}_1}-\mathcal{E}_{{p}_2}-\mathcal{E}_{{p}_3})[(1+f(t,r,{p}_1))f(t,r,{p}_2)f(t,r,{p}_3)-\\
&-f(t,r,{p}_1)(1+f(t,r,{p}_2))(1+f(t,r,{p}_3))]dp_1d{p}_2d{p}_3\Phi(t,r).
\end{aligned}
\end{equation}
This is done based on the fact that $g$ is also the principle small parameter used in the derivation of the system. The derivation starts with the usual Heisenberg equation of motion for the quantum field
operator. Then, by averaging the Heisenberg equation with respect to a broken-symmetry nonequilibrium ensemble, the Gross-Pitaevski equation for the condensate wavefunction follows. The quantum Boltzmann equation is derived by taking the difference between the Heisenberg equation and the  equation for the condensate wavefunction and keeping only the terms of low orders with respect to $g$. Therefore, in \eqref{GP}, one can drop \eqref{GP2} since it is of second order in $g$. 

To simplify the notations, let us omit $\hbar$, $m$ and $g$  and study the following kinetic-Schr\"odinger system
\begin{eqnarray}\label{System01}
\frac{\partial f}{\partial t}(t,r,p)+{p}\cdot\nabla_{{r}} f(t,r,p)  
& = & L[f] \ :=\ N_c(t,r)\left[\mathfrak{E}(p)\rho[f](t,r) \ - \ f(t,r,p)\right],\\\nonumber
&&\ \ \  \mbox{ on } \mathbb{R}_+\times\mathbb{R}^3\times\mathbb{R}^3,\\\nonumber
f(0,r,p)\ &=& \ f_0(r,p), (r,p)\in\mathbb{R}^3\times\mathbb{R}^3,\\\label{System02}
i  \frac{\partial \Phi(t,r)}{\partial t} \ &=& \ \Big(-{ \Delta_{{r}}}  +  |\Phi(t,r)|^2+\rho(t,r)\Big)\Phi(t,r),\\\nonumber
& &\ \ \ \ \ \ \ \ \ \mbox{ on }\mathbb{R}_+\times\mathbb{R}^3,\\\nonumber
\Phi(0,r)&=&\Phi_0(r), \forall r\in\mathbb{R}^3,\\\label{System03}
\rho[f](t,r)&=&\int_{\mathbb{R}^3}f(t,r,p')dp'.
\end{eqnarray}
Now, by putting $\Phi=e^{-i(1+M_{f_0})t}\Psi$, we obtain the  system \eqref{System1}-\eqref{System5}. Note that the reduced system \eqref{System1}-\eqref{System5} has the following conservation of mass:
\begin{equation}
\int_{\mathbb{R}^3}f(t,r,p)drdp \ = \ \int_{\mathbb{R}^3}f_0(r,p)drdp,
\end{equation}
and
\begin{equation}
\int_{\mathbb{R}^3}|\Psi(t,r)|^2drdp \ = \ \int_{\mathbb{R}^3}|\Psi_0(r)|^2drdp.
\end{equation}
  \section{The linear quantum Boltzmann equation}\label{Sec:Kinetic}
Let us consider the kinetic equation \eqref{System1}, with $N_c(t,r)$ being a given coefficient. We have that
$$\frac{d}{dt}\int_{\mathbb{R}^3\times\mathbb{R}^3}f(t,r,p)drdp =  0,$$
which implies
$$\int_{\mathbb{R}^3\times\mathbb{R}^3}f(t,r,p)drdp=\int_{\mathbb{R}^3\times\mathbb{R}^3}f(0,r,p)drdp = M_{f_0}, ~~~\forall t\in\mathbb{R}_+.$$
In the  following two subsections, we will consider two different scenarios of $L$: \eqref{Operator1} and \eqref{Operator2}. We will see that for the first case, the convergence rate to equilibrium is polynomial and for the second case, it is exponential.
\subsection{The decay rates when $L=L_1$}
We first observe that the following identities hold true
\begin{eqnarray}\label{IdentityKinetic1}
&& \frac{d}{dt}\int_{\mathbb{R}^3\times\mathbb{R}^3}|f(t,r,p)|^2\mathfrak{E}^{-1}(p)drdp\\\nonumber
& = & -\frac{1}{2}\int_{\mathbb{R}^3\times\mathbb{R}^3} N_c(t,r)\left[\mathfrak{E}(p)\int_{\mathbb{R}^3}f(t,r,p')dp'-f(t,r,p)\right]^2\mathfrak{E}^{-1}(p)drdp\\\nonumber
& = & -\frac{1}{2}\int_{\mathbb{R}^3\times\mathbb{R}^3\times\mathbb{R}^3} N_c(t,r)\left[\frac{f(t,r,p)}{\mathfrak{E}(p)}-\frac{f(t,r,p')}{\mathfrak{E}(p')}\right]^2\mathfrak{E}(p')\mathfrak{E}(p)dp'dpdr,
\end{eqnarray}
and
\begin{eqnarray}\label{IdentityKinetic2}
&&\frac{d}{dt}\int_{\mathbb{R}^3\times\mathbb{R}^3}|f(t,r,p)-f_{\infty}(p)|^2\mathfrak{E}^{-1}(p)drdp\\\nonumber
 & = &-\frac{1}{2}\int_{\mathbb{R}^3\times\mathbb{R}^3} N_c(t,r)\left[\mathfrak{E}(p)\int_{\mathbb{R}^3}(f(t,r,p')-f_\infty(p'))dp'-(f(t,r,p)-f_\infty(p))\right]^2\mathfrak{E}^{-1}(p)drdp\\\nonumber
& = & -\frac{1}{2}\int_{\mathbb{R}^3\times\mathbb{R}^3\times\mathbb{R}^3} N_c(t,r)\left[\frac{f(t,r,p)-f_{\infty}(p)}{\mathfrak{E}(p)}-\frac{f(t,r,p')-f_{\infty}(p')}{\mathfrak{E}(p')}\right]^2\mathfrak{E}(p')\mathfrak{E}(p)dp'dpdr. 
\end{eqnarray}
\begin{proposition}\label{PropNoE:Kinetic} Suppose that $f_0$ be a positive function in $L^1(\mathbb{R}^3\times\mathbb{R}^3)\cap \mathcal{L}$ and $N_c(t,r)$ is bounded from above by $\overline{C}_N>0$ and from below by $\underline{C}_N>0$. Under the assumption that $L=L_1$, Equation \eqref{System1}-\eqref{System1b} has a unique positive solution $f$, which decays polynomially in time towards the equilibrium $f_\infty$  in the following sense: there exists $\mathfrak{C}_1>0$  depending on $\|f_0-f_\infty\|_{\mathcal{L}},\underline{C}_N,\overline{C}_N$, such that
\begin{equation}
\label{PropNoE:Kinetic:1}
\|f(t)-f_\infty\|_{\mathcal{L}}\le \frac{\mathfrak{C}_1\big(\|f_0-f_\infty\|_{\mathcal{L}},\underline{C}_N,\overline{C}_N\big)}{\sqrt{1+t}}. 
\end{equation}
Moreover, there also exists $\mathfrak{C}_2>0$  depending on $\|f_0-f_\infty\|_{\mathcal{L}},\underline{C}_N,\overline{C}_N$, such that
\begin{equation}
\label{PropNoE:Kinetic:2}
\|\rho[f](t)-M_{f_0}\|_{L^2_r(\mathbb{R}^3)}\le \frac{\mathfrak{C}_2\big(\|f_0-f_\infty\|_{\mathcal{L}},\underline{C}_N,\overline{C}_N\big)}{\sqrt{1+t}}. 
\end{equation}
\end{proposition}
\begin{proof}
The existence and uniqueness result of the equation \eqref{System1} is classical  due to the same argument used in (cf. \cite{Fitzgibbon:1983:IBV}). 
\\ We now try to prove the decay rate \eqref{PropNoE:Kinetic:1} by assuming without loss of generality that $f_\infty=0$. Let us start with the following a priori estimate  by multiplying both sides of \eqref{System1} with $\mathrm{sign}f$:
\begin{equation}\label{PropNoE:Kinetic:E01}
\begin{aligned}
\frac{\partial |f|}{\partial t}(t,r,p)+{p}\cdot\nabla_{{r}} |f|(t,r,p)\ = &\  N_c(t,r)\left[\mathfrak{E}(p)\rho[f](t,r)\mathrm{sign}f(t,r,p) \ - \ |f|(t,r,p)\right] \\
\ \le & \ N_c(t,r)\left[\mathfrak{E}(p)\rho[|f|](t,r) \ - \ |f|(t,r,p)\right].
\end{aligned}
\end{equation}
Integrating both sides of Inequality \eqref{PropNoE:Kinetic:E01} yields
\begin{equation*}
\begin{aligned}
\frac{d }{d t}\int_{\mathbb{R}^3}|f|(t,r,p)drdp \ \leq \ 0, 
\end{aligned}
\end{equation*}
which implies
\begin{equation}\label{PropNoE:Kinetic:E02}
\int_{\mathbb{R}^3}|f(t,r,p)|dr dp \ \leq \ \int_{\mathbb{R}^3}|f_0(t,r,p)|dr dp=: M_{|f_0|}. 
\end{equation}
Now, taking the Fourier transform both sides of \eqref{System1}, we find
\begin{equation}\label{PropNoE:Kinetic:E1}
\partial_t \hat{f}(t,\zeta,p)\ + \ i(p\cdot \zeta)\hat{f}(t,\zeta,p) \ = \ \hat{N}_c(t,\zeta)*[\hat{\rho}(t,\zeta)\mathfrak{E}(p) \ - \ \hat{f}(t,\zeta,p)].
\end{equation}
Following the perturbed  energy estimate strategy introduced in \cite{DolbeaultMouhotSchmeiser:HFL:2015,BeauchardZuazua:2009:SCR}, we define 
\begin{equation}\label{PropNoE:Kinetic:E2}
\mathcal{E}[f](t,\zeta) \ = \ \left(\int_{\mathbb{R}^3}|\hat{f}|^2\mathfrak{E}^{-1}dp\right)(t,\zeta) \ + \ \delta \ \mathrm{Re}\left(R_\zeta [\hat{f}]\overline{\hat{f}}\mathfrak{E}^{-1}dp\right),
\end{equation}
where 
\begin{equation}\label{PropNoE:Kinetic:E3}
R_\zeta [\hat{f}] \ := \ \frac{-i\zeta}{1+|\zeta|^2}\rho(p\hat{f})\mathfrak{E},~~~\rho(p\hat{f}) \ = \ \int_{\mathbb{R}^3}p\hat{f}dp.
\end{equation}
Let us estimate the norm of $R_\zeta [\hat{f}]$, by using H\"older inequality for $\rho(v\hat{f})$
\begin{equation}\label{PropNoE:Kinetic:E4}
\begin{aligned}
\int_{\mathbb{R}^3}|R_\zeta [\hat{f}]|^2\mathfrak{E}^{-1}dp& \ = \ \int_{\mathbb{R}^3}\frac{|\zeta|^2}{(1+|\zeta|^2)^2}\left|\int_{\mathbb{R}^3}p\hat{f}dp\right|^2\mathfrak{E}dp\\
&\ \le \ \int_{\mathbb{R}^3}\frac{|\zeta|^2}{(1+|\zeta|^2)^2}\left(\int_{\mathbb{R}^3}|\hat{f}|^2\mathfrak{E}^{-1}dp\right)\left(\int_{\mathbb{R}^3}|p|^2\mathfrak{E}dp\right)\mathfrak{E}dp,
\end{aligned}
\end{equation}
which, due to the facts that the integral on $\mathbb{R}^3$ of $|p|^2\mathfrak{E}$ is finite and the inequality $|\zeta|^2\le \frac{1}{4}(1+|\zeta|^2)^2$, implies
\begin{equation}\label{PropNoE:Kinetic:E5}
\begin{aligned}
\int_{\mathbb{R}^3}|R_\zeta [\hat{f}]|^2\mathfrak{E}^{-1}dp
&\ \le \ C\int_{\mathbb{R}^3}\left(\int_{\mathbb{R}^3}|\hat{f}|^2\mathfrak{E}^{-1}dp\right)\mathfrak{E}dp\\
&\ \le \ C\left(\int_{\mathbb{R}^3}|\hat{f}|^2\mathfrak{E}^{-1}dp\right),
\end{aligned}
\end{equation}
where the last inequality follows from the fact that the integral on $\mathbb{R}^3$ of $\mathfrak{E}$ is finite.
From Inequality \eqref{PropNoE:Kinetic:E5}, we deduce that, for $\delta$ small enough, there exist two positive constants $C_1$ and $C_2$ independent of $\zeta$ and $t$ such that
\begin{equation}\label{PropNoE:Kinetic:E6}
C_1\left(\int_{\mathbb{R}^3}|\hat{f}|^2\mathfrak{E}^{-1}dp\right)(t,\zeta)\ \le\ \mathcal{E}[f](t,\zeta)\ \le \ C_2\left(\int_{\mathbb{R}^3}|\hat{f}|^2\mathfrak{E}^{-1}dp\right)(t,\zeta).
\end{equation}
We  estimate the derivative in time of the norm $\mathcal{L}^2(\mathbb{R}^3\times\mathbb{R}^3)$ of the first term of $\mathcal{E}[f]$ in \eqref{PropNoE:Kinetic:E2}. It is straightforward that
\begin{equation}\label{PropNoE:Kinetic:E7}
\begin{aligned}
\partial_t\int_{\mathbb{R}^3\times \mathbb{R}^3}|{f}|^2\mathfrak{E}^{-1}drdp &=~~~2\int_{\mathbb{R}^3\times \mathbb{R}^3}(\partial_t{f}){{f}}\mathfrak{E}^{-1}dxdp.
\end{aligned}
\end{equation}
Using Equation \eqref{System1} to replace $\partial_t f$ a in the above equation by $-{p}\cdot\nabla_{{r}} f \ + \ L[f]$ yields
\begin{equation}\label{PropNoE:Kinetic:E8}
\begin{aligned}
\partial_t\int_{\mathbb{R}^3\times \mathbb{R}^3}|{f}|^2\mathfrak{E}^{-1}drdp \ \ &=~~~2\int_{\mathbb{R}^3\times \mathbb{R}^3}\left(-{p}\cdot\nabla_{{r}} f \ + \ L[f]\right){f}\mathfrak{E}^{-1}drdp\\ 
\ \ &=~~~2\int_{\mathbb{R}^3\times \mathbb{R}^3}L[f]{f}\mathfrak{E}^{-1}drdp\\
\ \ &\le~~~ - 2\int_{\mathbb{R}^3\times \mathbb{R}^3}N_c\left[\mathfrak{E}\rho [f]\ - \ f\right]^2\mathfrak{E}^{-1}drdp.
\end{aligned}
\end{equation}
Using the fact that $N_c\geq \underline{C}_\Phi$, we can bound the integral of  $N_c\left[\mathfrak{E}(p)\rho(t,r) \ - \ f(t,r,p)\right]^2\mathfrak{E}^{-1}$ in the above inequality as 
\begin{equation}\label{PropNoE:Kinetic:E9}
\begin{aligned}
\partial_t\int_{\mathbb{R}^3\times \mathbb{R}^3}|{f}|^2\mathfrak{E}^{-1}drdp \ \ &\le~~~ - \underline{C}_\Phi\int_{\mathbb{R}^3\times \mathbb{R}^3}\left[\mathfrak{E}\rho[f] \ - \ f\right]^2\mathfrak{E}^{-1}drdp.
\end{aligned}
\end{equation}
Notice that in the above inequality, by the Parseval identity, we can switch the integral in $r$ into an integral in $\zeta$, which yields
\begin{equation}\label{PropNoE:Kinetic:E9}
\begin{aligned}
\partial_t\int_{\mathbb{R}^3\times \mathbb{R}^3}|\hat{f}|^2\mathfrak{E}^{-1}d\zeta dp \ \ &\le~~~ - \underline{C}_\Phi\int_{\mathbb{R}^3\times \mathbb{R}^3}\left[\mathfrak{E}\hat\rho[f] \ - \ \hat{f}\right]^2\mathfrak{E}^{-1}d\zeta dp.
\end{aligned}
\end{equation}
We now estimate the derivative in time of the norm $\mathcal{L}^2(\mathbb{R}^3\times\mathbb{R}^3)$  of the second term of $\mathcal{E}[f]$ in \eqref{PropNoE:Kinetic:E2}. Observe that 
\begin{equation}\label{PropNoE:Kinetic:E10}
\begin{aligned}
\partial_t  \int_{\mathbb{R}^3}R_\zeta [\hat{f}]\overline{\hat{f}}\mathfrak{E}^{-1}dp \ \ = \ \  \int_{\mathbb{R}^3}R_\zeta [\partial_t  \hat{f}]\overline{\hat{f}}\mathfrak{E}^{-1}dp \ \ + \ \ \int_{\mathbb{R}^3}R_\zeta [\hat{f}]\partial_t \overline{\hat{f}}\mathfrak{E}^{-1}dp,
\end{aligned}
\end{equation}
Using again Equation \eqref{System1} to replace $\partial_t \hat{f}$ and $\partial_t\bar{\hat{f}}$ in the above equation by $-{p}\cdot \zeta \hat{f} \ + \ L[\hat{f}]$ and $-{p}\cdot \zeta \overline{\hat{f}} \ + \ L[\overline{\hat{f}}]$, we find
\begin{equation}\label{PropNoE:Kinetic:E11}
\partial_t  \int_{\mathbb{R}^3\times \mathbb{R}^3}R_\zeta [\hat{f}]\overline{\hat{f}}\mathfrak{E}^{-1}dp \ =  \ I \ = \ I_1 \  \ + \ \ I_2 \ \ + \ \ I_3 \ \ + \ \ I_4,
\end{equation}
where 
\begin{equation}\label{PropNoE:Kinetic:E12}
\begin{aligned}
I_1 ~~~~~ & := ~~~~~  - \int_{\mathbb{R}^3\times \mathbb{R}^3}R_\zeta [i p\cdot \zeta \   \hat{f}] \ \overline{\hat{f}} \ \mathfrak{E}^{-1} \ dpd\zeta,\\
I_2 ~~~~~ &:= ~~~~~ \int_{\mathbb{R}^3\times \mathbb{R}^3}R_\zeta [L[\hat{f}]]  \  \overline{\hat{f}} \  \mathfrak{E}^{-1} \ dpd\zeta, \\
I_3 ~~~~~ &:= ~~~~~ -\int_{\mathbb{R}^3\times \mathbb{R}^3}R_\zeta [\hat{f}] \  \overline{i p\cdot \zeta  \hat{f}} \  \mathfrak{E}^{-1} \ dpd\zeta, \\
I_4 ~~~~~ &:= ~~~~~ \int_{\mathbb{R}^3\times \mathbb{R}^3}R_\zeta [\hat{f}] \ \overline{L[\hat{f}]}\  \mathfrak{E}^{-1} \ dpd\zeta.
\end{aligned}
\end{equation}
In the sequel, we will estimate $I_1$, $I_2$, $I_3$ and $I_4$ step by step. Let us start with $I_1$:
\begin{equation}\label{PropNoE:Kinetic:E13}
\begin{aligned}
I_1 \   & := \  - \int_{\mathbb{R}^3\times \mathbb{R}^3}\frac{-i\zeta}{1+|\zeta|^2} \ \rho(pip\cdot \zeta\hat{f} )  \ \overline{\hat{f}} \ \mathfrak{E} \ \mathfrak{E}^{-1} \ dpd\zeta\\
\ &= \ - \int_{\mathbb{R}^3\times \mathbb{R}^3}\frac{\zeta \otimes \zeta }{1+|\zeta|^2} \ : \ \rho[p \otimes p \hat{f}] \ \overline{\hat{f}} \ dpd\zeta\\
\ &= \ - \int_{\mathbb{R}^3}\frac{\zeta \otimes \zeta }{1+|\zeta|^2} \ : \ \rho[p \otimes p \hat{f}] \ \overline{\rho[\hat{f}]}d\zeta,
\end{aligned}
\end{equation}
in which the following notation for matrix contraction has been used
$$\begin{bmatrix}
    a_{11}       & a_{12} & a_{13} \\
    a_{21}       & a_{22} & a_{23} \\
    a_{31}       & a_{32} & a_{33} 
\end{bmatrix}
 \ : \ \begin{bmatrix}
    b_{11}       & b_{12} & b_{13} \\
    b_{21}       & b_{22} & b_{23} \\
    b_{31}       & b_{32} & b_{33} 
\end{bmatrix} \ = \ \sum_{i,j=1}^3 a_{i,j}b_{i,j}.$$
In Equation \eqref{PropNoE:Kinetic:E13}, we split $ p \otimes p \hat{f} $ as the sum of $ p \otimes p \mathfrak{E} \rho[\hat{f}]$ and $ p \otimes p (\hat{f} -  \mathfrak{E}  \rho[\hat{f}])$, and obtain  
\begin{equation}\label{PropNoE:Kinetic:E14}
\begin{aligned}
I_1 \   & := \  I_{11}  \ + \ I_{12},
\end{aligned}
\end{equation}
where
\begin{equation}\label{PropNoE:Kinetic:E15}
\begin{aligned}
I_{11} \   & :=  \ - \int_{\mathbb{R}^3}\frac{\zeta \otimes \zeta }{1+|\zeta|^2} \ : \ \rho[p \otimes  p \mathfrak{E}  \rho[\hat{f}]] \ \overline{\rho[\hat{f}]} d\zeta\ = \  - \int_{\mathbb{R}^3}\frac{\zeta \otimes \zeta }{1+|\zeta|^2} \ : \ \rho[p \otimes p \mathfrak{E} ] \   \left|\rho[\hat{f}]\right|^2d\zeta,\\
I_{12} \   & :=  \ - \int_{\mathbb{R}^3}\frac{\zeta \otimes \zeta }{1+|\zeta|^2} \ : \ \rho[p \otimes p (\hat{f}-\mathfrak{E} \rho[\hat{f}])] \ \overline{\rho[\hat{f}]} d\zeta.
\end{aligned}
\end{equation}
Now, for $I_{11}$, the fact that $\rho[p \otimes p \mathfrak{E} ]=\mathrm{Id}$ implies 
\begin{equation}\label{PropNoE:Kinetic:E16}
I_{11} \     =  \ - \int_{\mathbb{R}^3} \frac{\zeta \otimes \zeta }{1+|\zeta|^2} \ : \ \mathrm{Id} \   \left|\rho[\hat{f}]\right|^2 \  d\zeta   =  \int_{\mathbb{R}^3} \ - \frac{|\zeta|^2}{1+|\zeta|^2}  \left|\rho[\hat{f}]\right|^2 d\zeta.
\end{equation}
Set $F=\hat{f}-\mathfrak{E} \rho[\hat{f}]$,  the second term $I_{12}$ can be estimated as follows
\begin{equation}\label{PropNoE:Kinetic:E17}
\begin{aligned}
|I_{12}| \   & =  \  \int_{\mathbb{R}^3} \left|\frac{\zeta \otimes \zeta }{1+|\zeta|^2} \ : \ \rho[p \otimes p F] \ \overline{\rho[\hat{f}]}\right| d\zeta\\
\ & \le  \ \int_{\mathbb{R}^3}  \frac{|\zeta|^2\left|{\rho[\hat{f}]}\right|}{1+|\zeta|^2} \ \int_{\mathbb{R}^3} |p|^2  |F| \ dp \,d\zeta
\end{aligned}
\end{equation}
which, by H\"older inequality applied to the integral on $p$, can be bounded as
\begin{equation}\label{PropNoE:Kinetic:E18}
\begin{aligned}
|I_{12}| \  
\ & \le  \ \int_{\mathbb{R}^3} \frac{|\zeta|^2\left|{\rho[\hat{f}]}\right|}{1+|\zeta|^2} \ \left(\int_{\mathbb{R}^3}  |F|^2\mathfrak{E}^{-1} \ dp\right)^{\frac{1}{2}}\ \left(\int_{\mathbb{R}^3}  |p|^4\mathfrak{E}\ dp\right)^{\frac{1}{2}} d\zeta\\
\ & \le \ C \int_{\mathbb{R}^3}  \frac{|\zeta|^2}{1+|\zeta|^2} \ \left(\int_{\mathbb{R}^3}  |F|^2\mathfrak{E}^{-1} \ dp\right)^{\frac{1}{2}}\left|{\rho[\hat{f}]}\right|d\zeta.
\end{aligned}
\end{equation}
Combining \eqref{PropNoE:Kinetic:E16} and \eqref{PropNoE:Kinetic:E18} yields the following estimate on $I_1$
\begin{equation}\label{PropNoE:Kinetic:EstimateI1}
I_{1} \     \le   \ - \int_{\mathbb{R}^3}  \frac{|\zeta|^2}{1+|\zeta|^2}  \left|\rho[\hat{f}]\right|^2 \ d\zeta + \ \int_{\mathbb{R}^3} C\frac{|\zeta|^2}{1+|\zeta|^2} \ \left(\int_{\mathbb{R}^3}  |F|^2\mathfrak{E}^{-1} \ dp\right)^{\frac{1}{2}}\left|{\rho[\hat{f}]}\right| d\zeta.
\end{equation}
We continue with estimating the second term $I_2$, which could be written under the following form
\begin{equation}\label{PropNoE:Kinetic:E19}
\begin{aligned}
I_2 ~~~~~ &= ~~~~~ \int_{\mathbb{R}^3\times \mathbb{R}^3} \frac{-i\zeta}{1+|\zeta|^2} \cdot \rho(p L[\hat{f}])  \  \overline{\hat{f}} \ \mathfrak{E}  \  \mathfrak{E}^{-1} \ dp d\zeta\\
&= ~~~~~ \int_{\mathbb{R}^3\times \mathbb{R}^3} \frac{i \zeta}{1+|\zeta|^2} \cdot \rho(p L[\hat{f}])  \  \overline{\hat{f}} \ dp d\zeta.
\end{aligned}
\end{equation}
It is straightforward that $L(\mathfrak{E})=0$, which implies  $$L[\rho[\hat{f}]\mathfrak{E}]=L(\mathfrak{E})\rho[\hat{f}]=0.$$
Hence $L[F]=L[\hat{f}]$ and it follows from \eqref{PropNoE:Kinetic:E19} that
\begin{equation}\label{PropNoE:Kinetic:E20}
\begin{aligned}
I_2 ~~~&= ~~~ \int_{\mathbb{R}^3\times \mathbb{R}^3} \frac{i \zeta}{1+|\zeta|^2} \cdot \rho[p L[F]]  \  \overline{\hat{f}} \ dp d\zeta
\ = \ \int_{\mathbb{R}^3} \frac{i \zeta}{1+|\zeta|^2} \cdot \rho[p L[F]]\rho[\overline{\hat{f}}] d\zeta.
\end{aligned}
\end{equation}
Let us estimate  $\rho[p L[F]] $ first. By definition this term  can be rewritten as
\begin{eqnarray*}
\rho[p L[F]] (t,\zeta)\ & =  & \ \left(\int_{\mathbb{R}^3}p\hat{N}_c(t,\cdot)*[\hat{\rho}[F](t,\cdot)\mathfrak{E}(p) \ - \ F(t,\cdot,p)]\ dp\right)(\zeta)\\
\ & = &\ \left(\hat{N}_c(t,\cdot)*\hat{\rho}[F](t,\cdot) \int_{\mathbb{R}^3}p\mathfrak{E}(p) \ dp \ - \ \hat{N}_c(t,\cdot)* \ \int_{\mathbb{R}^3}p F(t,\cdot,p)\ dp \right)(\zeta).
\end{eqnarray*}
Since $\int_{\mathbb{R}^3}p\mathfrak{E}(p)dp=0$, the first term in  $\rho[pL[F]]$ is zero and $\rho[pL[F]]$  can be reduced to
\begin{eqnarray*}
\rho[p L[F]](t,\zeta) \ & =  & \  \ -\left(\hat{N}_c(t,\cdot)* \ \int_{\mathbb{R}^3} p F(t,\cdot,p)\ dp\right)(\zeta),
\end{eqnarray*}
which, by H\"older inequality applied to the integral in $p$, can be bounded as
\begin{eqnarray*}
\left|\rho[p L[F]](t,\zeta)\right| \ & \le   & \  \ \left| \ \left(\int_{\mathbb{R}^3}|\hat{N}_c(t,\cdot)* F(t,\cdot,p)|^2\mathfrak{E}(p)^{-1}\ dp\right)^{\frac{1}{2}}(\zeta)\left(\int_{\mathbb{R}^3} |p|^2\mathfrak{E}(p)\ dp\right)^{\frac{1}{2}}\right|\\
\ &\le &\ C\left(\int_{\mathbb{R}^3} |\hat{N}_c(t,\cdot)* F(t,\cdot,p)|^2\mathfrak{E}(p)^{-1}\ dp\right)^{\frac{1}{2}}(\zeta),
\end{eqnarray*}
which could be bounded by using H\"oder inequality for the integral in $p$ of $\rho[p L[F]] $ as
\begin{equation}\label{PropNoE:Kinetic:EstimateI2}
\begin{aligned}
|I_2|~~~&\le  ~~~ \ C\int_{\mathbb{R}^3}\frac{|\zeta|}{1+|\zeta|^2}\ \left(\int_{\mathbb{R}^3} |\hat{N}_c(t,\cdot)* F(t,\cdot,p)|^2\mathfrak{E}(p)^{-1}\ dp\right)^{\frac{1}{2}}(\zeta)\left|\rho[{\hat{f}}]\right| d\zeta.
\end{aligned}
\end{equation}
Now, we estimate $I_3$
\begin{equation}\label{PropNoE:Kinetic:E21}
\begin{aligned}
I_3 ~~~~~ &= ~~~~~ - \int_{\mathbb{R}^3\times \mathbb{R}^3}\frac{-i\zeta}{1+|\zeta|^2}\cdot \rho(p\hat{f})\mathfrak{E} \ i p\cdot \zeta  \overline{  \hat{f}} \  \mathfrak{E}^{-1} \ dpd\zeta \\
~~~~~ &= ~~~~~ - \int_{\mathbb{R}^3} \frac{|\zeta\cdot \rho[p\hat{f}]|^2}{1+\zeta^2} d\zeta \\
~~~~~ &= ~~~~~ - \int_{\mathbb{R}^3} \frac{|\zeta\cdot \rho[pF]|^2}{1+\zeta^2} d\zeta ,
\end{aligned}
\end{equation}
where we have used the fact that 
$$\rho[p \rho[\hat{f}] \mathfrak{E}] \ = \ \rho[\hat{f}] \ \int_{\mathbb{R}^3} \ p \mathfrak{E} \ dp \ = \ 0.$$
In order to estimate $|I_3|$, we will first try to bound $\rho[p F]$. By H\"older inequality, we find
\begin{eqnarray*}
|\rho[p F](t,\zeta)| \ & = & \ \left|\int_{\mathbb{R}^3}|p|F(t,\zeta,p)dp \right| \\
\ &\le & \ \left(\int_{\mathbb{R}^3}|F(t,\zeta,p)|^2\mathfrak{E}^{-1}(p)dp \right)^{\frac{1}{2}}\left(\int_{\mathbb{R}^3}|p|^2\mathfrak{E}(p)dp \right)^{\frac{1}{2}}\\
\ &\le & \ C\left(\int_{\mathbb{R}^3}|F(t,\zeta,p)|^2\mathfrak{E}^{-1}(p)dp \right)^{\frac{1}{2}},
\end{eqnarray*}
which, together with Inequality \eqref{PropNoE:Kinetic:E21}, implies
\begin{equation}\label{PropNoE:Kinetic:EstimateI3}
\begin{aligned}
|I_3| ~~~ &\le  ~~~ C \int_{\mathbb{R}^3\times \mathbb{R}^3} \frac{\zeta^2}{1+\zeta^2} \left(\int_{\mathbb{R}^3}|F(t,\zeta,p)|^2\mathfrak{E}^{-1}(p)dp \right) d\zeta.
\end{aligned}
\end{equation}
Estimating $I_4$ is quite easy and we proceed as follows:
\begin{equation}\label{PropNoE:Kinetic:EstimateI4}
\begin{aligned}
I_4 \ & = \  \int_{\mathbb{R}^3\times \mathbb{R}^3}\frac{-i\zeta}{1+|\zeta|^2}\rho(p\hat{f})\mathfrak{E} \ \overline{L[\hat{f}]}\  \mathfrak{E}^{-1} \ dpd\zeta \\
& = \  \frac{-i\zeta}{1+|\zeta|^2}\rho(p\hat{f})\int_{\mathbb{R}^3\times \mathbb{R}^3} \ \overline{L[\hat{f}]}\ dpd\zeta \\
&= \ 0.
\end{aligned}
\end{equation}
Putting together the four inequalities \eqref{PropNoE:Kinetic:EstimateI1}, \eqref{PropNoE:Kinetic:EstimateI2}, \eqref{PropNoE:Kinetic:EstimateI3} and \eqref{PropNoE:Kinetic:EstimateI4} yields the following estimate on $I$ 
\begin{equation}\label{PropNoE:Kinetic:E22}\begin{aligned}
I \     \le &    \ - \int_{\mathbb{R}^3}  \frac{|\zeta|^2}{1+|\zeta|^2}  \left|\rho[\hat{f}]\right|^2 \ d\zeta \\
& \ + \int_{\mathbb{R}^3}\frac{|\zeta|}{1+|\zeta|^2}\left|\ \left(\int_{\mathbb{R}^3} |\hat{N}_c(t,\cdot)* F(t,\cdot,p)|^2(\zeta)\mathfrak{E}(p)^{-1}\ dp\right)^{\frac{1}{2}}\right|\left|\rho[{\hat{f}}]\right| d\zeta\\
& \ + C \int_{\mathbb{R}^3\times\mathbb{R}^3} \frac{\zeta^2}{1+\zeta^2} \int_{\mathbb{R}^3}|F(t,\zeta,p)|^2\mathfrak{E}^{-1}(p)dp  d\zeta,
\end{aligned}
\end{equation}
where $C$ is a constant varying from lines to lines.\\
Applying the Cauchy-Schwarz inequality 
$$|\alpha\beta|\le \frac{\alpha^2}{2\epsilon}\ + \ \frac{\epsilon}{2}\beta^2$$
to the right hand side of \eqref{PropNoE:Kinetic:E22}, we find
\begin{equation}\label{PropNoE:Kinetic:E23}\begin{aligned}
I \     \le &   \ - \int_{\mathbb{R}^3}  \frac{ |\zeta|^2}{1+|\zeta|^2}  \left|\rho[\hat{f}]\right|^2 \ d\zeta \\
& + \int_{\mathbb{R}^3}\frac{\epsilon |\zeta|^2}{(1+|\zeta|^2)^2}\left|\rho[{\hat{f}}]\right|^2d\zeta \ +\ \frac{C}{\epsilon} \int_{\mathbb{R}^3\times\mathbb{R}^3 } |\hat{N}_c(t,\cdot)* F(t,\cdot,p)|^2(
\zeta)\mathfrak{E}(p)^{-1}\  dpd\zeta\\
& + \ {C} \int_{\mathbb{R}^3\times\mathbb{R}^3 }  \frac{ |\zeta|^2}{1+|\zeta|^2} |F|^2\mathfrak{E}^{-1} \ dpd\zeta\\
\le &  \ C \int_{\mathbb{R}^3} \left(-\frac{ |\zeta|^2}{1+|\zeta|^2} +\frac{\epsilon |\zeta|^2}{(1+|\zeta|^2)^2}\right) \left|\rho[\hat{f}]\right|^2 \ d\zeta \ \\
& + {C}\left(1+\frac{1}{\epsilon}\right)\int_{\mathbb{R}^3\times\mathbb{R}^3 } \frac{ |\zeta|^2}{1+|\zeta|^2} |F|^2\mathfrak{E}^{-1} \ dpd\zeta + \frac{C}{\epsilon} \int_{\mathbb{R}^3\times\mathbb{R}^3 } |\hat{N}_c(t,\zeta)* F(t,\zeta,p)|^2\mathfrak{E}(p)^{-1}\  dpd\zeta.
\end{aligned}
\end{equation}
Let $\check{F}$ be the inverse Fourier transform of $F$, by Parseval identity, we have
\begin{eqnarray*}
\int_{\mathbb{R}^3} |\hat{N}_c(t,\cdot)* F(t,\cdot,p)|^2(\zeta)\mathfrak{E}(p)^{-1}\  d\zeta &  = & \int_{\mathbb{R}^3} |{N}_c(t,r)|^2 |\check{F}(t,r,p)|^2\mathfrak{E}(p)^{-1}\  dr.
\end{eqnarray*}
Using the assumption $|N_c(t,r)|\le \overline{C}_\Phi$, the right hand side of the above identity can be estimated as
\begin{eqnarray*}
\int_{\mathbb{R}^3} |\hat{N}_c(t,\cdot)* F(t,\cdot,p)|^2(\zeta)\mathfrak{E}(p)^{-1}\  d\zeta &  \le  & |\overline{C}_\Phi|^2 \int_{\mathbb{R}^3} |\check{F}(t,r,p)|^2\mathfrak{E}(p)^{-1}\  dr.
\end{eqnarray*}
Applying Parseval identity to the right hand side of the above inequality leads to
\begin{eqnarray*}
\int_{\mathbb{R}^3}  |\hat{N}_c(t,\cdot)* F(t,\cdot,p)|^2(\zeta)\mathfrak{E}(p)^{-1}\  d\zeta &  \le  & |\overline{C}_\Phi|^2 \int_{\mathbb{R}^3} |{F}(t,\zeta,p)|^2\mathfrak{E}(p)^{-1}\  d\zeta,
\end{eqnarray*}
which, together with \eqref{PropNoE:Kinetic:E23}, implies 
\begin{equation}\label{PropNoE:Kinetic:E24}\begin{aligned}
I \     \le &  \ C \int_{\mathbb{R}^3} \left(-\frac{ |\zeta|^2}{1+|\zeta|^2} +\frac{\epsilon |\zeta|^2}{(1+|\zeta|^2)^2}\right) \left|\rho[\hat{f}]\right|^2 \ d\zeta \ +\\
\  & \ +  \ {C}\left(1+\frac{1}{\epsilon}\right)\int_{\mathbb{R}^3\times\mathbb{R}^3 }  \frac{ |\zeta|^2}{1+|\zeta|^2} |F|^2\mathfrak{E}^{-1} \ dpd\zeta.
\end{aligned}
\end{equation}
We now combine \eqref{PropNoE:Kinetic:E8} and \eqref{PropNoE:Kinetic:E24}, to get the following estimate on $\mathcal{E}$
\begin{equation}\label{PropNoE:Kinetic:E25}\begin{aligned}
\partial_t\int_{\mathbb{R}^3}\mathcal{E}[f](t,\zeta)d\zeta \     \le &  \ C\delta \int_{\mathbb{R}^3} \left(-\frac{ |\zeta|^2}{1+|\zeta|^2} +\frac{\epsilon |\zeta|^2}{(1+|\zeta|^2)^2}\right) \left|\rho[\hat{f}]\right|^2 \ d\zeta \ +\\
\  & \ +  \ \left[-1+{C\delta}\left(1+\frac{1}{\epsilon}\right)\right]\int_{\mathbb{R}^3\times\mathbb{R}^3 }  |F|^2\mathfrak{E}^{-1} \ dpd\zeta.
\end{aligned}
\end{equation}
Choosing $\delta$ and $\epsilon$, such that 
$$-\delta\frac{ |\zeta|^2}{1+|\zeta|^2} +\frac{\delta\epsilon |\zeta|^2}{(1+|\zeta|^2)^2}\le -\frac{\delta}{2}\frac{ |\zeta|^2}{1+|\zeta|^2},$$
and
$$-1+{\delta C}\left(1+\frac{1}{\epsilon}\right)\le \frac{1}{2},$$
we get the following estimate from \eqref{PropNoE:Kinetic:E25}
\begin{equation}\label{PropNoE:Kinetic:E26}\begin{aligned}
\partial_t\int_{\mathbb{R}^3}\mathcal{E}[f](t,\zeta)d\zeta \     \le &  \ - C\frac{\delta}{2} \int_{\mathbb{R}^3} \frac{ |\zeta|^2}{1+|\zeta|^2} \left|\rho[\hat{f}]\right|^2 \ d\zeta \ - \
  \ \frac{1}{2}\int_{\mathbb{R}^3\times\mathbb{R}^3 }  |F|^2\mathfrak{E}^{-1} \ dpd\zeta.
\end{aligned}
\end{equation}
Suppose that $C\delta <1$, we deduce from the above inequality that
\begin{equation}\label{PropNoE:Kinetic:E27}\begin{aligned}
\partial_t\int_{\mathbb{R}^3}\mathcal{E}[f](t,\zeta)d\zeta \     \le &  \ - C\frac{\delta}{2} \int_{\mathbb{R}^3} \frac{ |\zeta|^2}{1+|\zeta|^2} \left|\rho[\hat{f}]\right|^2 \ d\zeta \ - \
  \ C\frac{\delta}{2}\int_{\mathbb{R}^3\times\mathbb{R}^3 }  \frac{ |\zeta|^2}{1+|\zeta|^2}  |F|^2\mathfrak{E}^{-1} \ dpd\zeta.
\end{aligned}
\end{equation}
Using the identity
$$\int_{\mathbb{R}^3}\mathfrak{E}(p)dp=1,$$
we find that 
$$\int_{\mathbb{R}^3} \frac{ |\zeta|^2}{1+|\zeta|^2} \left|\rho[\hat{f}]\right|^2 \ d\zeta \ = \ \int_{\mathbb{R}^3\times\mathbb{R}^3} \frac{ |\zeta|^2}{1+|\zeta|^2} \left|\rho[\hat{f}]\mathfrak{E}\right|^2 \mathfrak{E}^{-1} \ d\zeta dp,$$
which, combining with \eqref{PropNoE:Kinetic:E26}, implies
\begin{equation}\label{PropNoE:Kinetic:E28}\begin{aligned}
\partial_t\int_{\mathbb{R}^3}\mathcal{E}[f](t,\zeta)d\zeta \     \le &  \ - C\frac{\delta}{2} \left\{\int_{\mathbb{R}^3\times\mathbb{R}^3} \frac{ |\zeta|^2}{1+|\zeta|^2} \left[\left|\rho[\hat{f}]\mathfrak{E}\right|^2 \ + \ \left|\hat{f}- \rho[\hat{f}]\mathfrak{E}\right|^2   \right]\mathfrak{E}^{-1}dpd\zeta\right\}.
\end{aligned}
\end{equation}
Let us remark that by Cauchy-Schwarz inequality, we have
$$\left|\rho[\hat{f}]\mathfrak{E}\right|^2 \ + \ \left|\hat{f}- \rho[\hat{f}]\mathfrak{E}\right|^2 \ge \frac{1}{2}\left|\hat{f}\right|^2,$$
which, together with \eqref{PropNoE:Kinetic:E28}, yields
\begin{equation}\label{PropNoE:Kinetic:E29}\begin{aligned}
\partial_t\int_{\mathbb{R}^3}\mathcal{E}[f](t,\zeta)d\zeta \     \le &  \ - C\frac{\delta}{4} \left[\int_{\mathbb{R}^3\times\mathbb{R}^3} \frac{ |\zeta|^2}{1+|\zeta|^2} \left|\hat{f}\right|^2 \mathfrak{E}^{-1} dpd\zeta\right] \ =: \ A[f].
\end{aligned}
\end{equation}
Let us estimate $A[f]$, by H\"older inequality
\begin{equation}\label{PropNoE:Kinetic:E30}
\left[\int_{\mathbb{R}^3\times\mathbb{R}^3} \frac{ |\zeta|^2}{1+|\zeta|^2} \left|\hat{f}\right|^2 \mathfrak{E}^{-1} dpd\zeta\right]\left[\int_{\mathbb{R}^3\times\mathbb{R}^3} \frac{1+|\zeta|^2}{ |\zeta|^2} \left|\hat{f}\right|^2 \mathfrak{E}^{-1} dpd\zeta\right]\ \ge \ \left[\int_{\mathbb{R}^3\times\mathbb{R}^3}  \left|\hat{f}\right|^2 \mathfrak{E}^{-1} dpd\zeta\right]^2.
\end{equation}
In order to obtain an inequality for $A[f]$, we will prove that the factor $$B[f]\ := \ \int_{\mathbb{R}^3\times\mathbb{R}^3} \frac{1+|\zeta|^2}{ |\zeta|^2} \left|\hat{f}\right|^2 \mathfrak{E}^{-1} dpd\zeta$$ in the above inequality is bounded. It is straightforward from Inequality  \eqref{PropNoE:Kinetic:E9}, that 
\begin{equation}\label{PropNoE:Kinetic:E31}
\begin{aligned}
B_1[f] \ := & \  \int_{\mathbb{R}^3\times\mathbb{R}^3}\left|\hat{f}(t,\zeta,p)\right|^2 \mathfrak{E}^{-1}(p) dpd\zeta \ \le \ \int_{\mathbb{R}^3\times\mathbb{R}^3}\left|\hat{f}_0(\zeta,p)\right|^2 \mathfrak{E}^{-1}(p) dpd\zeta \\
\ \le & \int_{\mathbb{R}^3\times\mathbb{R}^3}\left|{f}_0(r,p)\right|^2 \mathfrak{E}^{-1}(p) dpdr.
\end{aligned}
\end{equation}
We need only to estimate the quantity
$$B_2[f]\ := \ \int_{\mathbb{R}^3\times\mathbb{R}^3} \frac{\left|\hat{f}\right|^2}{ |\zeta|^2}  \mathfrak{E}^{-1} dpd\zeta,$$
which, by splitting the integral of $\zeta$ on $\mathbb{R}^3$ into the sum of two integrals on $\{|\zeta|\le 1\}$ and $\{|\zeta| >1\}$, could be rewritten as
\begin{equation}\label{PropNoE:Kinetic:E32}
B_2[f] \ = \ B_{21}[f] \ + \ B_{22}[f],
\end{equation}
where
$$B_{21}[f]\ := \ \int_{\{|\zeta|\le 1\}\times\mathbb{R}^3} \frac{\left|\hat{f}\right|^2}{ |\zeta|^2}  \mathfrak{E}^{-1} dpd\zeta, \ \ \ B_{22}[f] \ := \ \int_{\{|\zeta|> 1\}\times\mathbb{R}^3} \frac{\left|\hat{f}\right|^2}{ |\zeta|^2}  \mathfrak{E}^{-1} dpd\zeta.$$
The second term $B_{22}[f]$ can be bounded by $B_1[f]$ in a straightforward manner as follows
\begin{equation}\label{PropNoE:Kinetic:E33}
B_{22}[f] \ \le \ B_{1}[f] \ \le \ \int_{\mathbb{R}^3\times\mathbb{R}^3}\left|{f}_0(r,p)\right|^2 \mathfrak{E}^{-1}(p) dpdr.
\end{equation}
We estimate the first term $B_{21}[f]$
\begin{equation}\label{PropNoE:Kinetic:E34}
\begin{aligned}
B_{21}[f] \ \le & \   \int_{\mathbb{R}^3}  \int_{\{|\zeta|\le 1\}}\frac{\left|\hat{f}\right|^2}{ |\zeta|^2}  \mathfrak{E}^{-1} d\zeta dp \ \le \  \int_{\mathbb{R}^3} \sup_{|
\zeta|
\le 1} \left|\hat{f}(\zeta,p)\right|^2 \int_{\{|\zeta|\le 1\}}\frac{1}{ |\zeta|^2}  \mathfrak{E}^{-1} d\zeta dp\\
\ \le & \ C\int_{\mathbb{R}^3}\left(\int_{\mathbb{R}^3}|f(t,r,p)|dr\right)^2\mathfrak{E}^{-1}dp.
\end{aligned}
\end{equation}
In order to bound the right hand side of \eqref{PropNoE:Kinetic:E34}, let us define $$G(t,p) \ = \ \int_{\mathbb{R}^3}|f(t,r,p)|dr,$$
and then by \eqref{PropNoE:Kinetic:E01}
$$\partial_t G \ + \ \int_{\mathbb{R}^3}N_c(t,r)|f(t,r,p)dr| \ \le  \ \int_{\mathbb{R}^3}N_c(t,r) \mathfrak{E}(p) \rho[|f|](t,r)dr.$$
By using the bounds $\underline{C}_N \le N_c(t,r) \le \overline{C}_N$ and \eqref{PropNoE:Kinetic:E02}, we deduce from the above identity that
$$\partial_t G \ + \ \underline{C}_N  G \ \le \ M_{|f_0|}\overline{C}_N\mathfrak{E}(p).$$
Multiplying the above inequality by $G\mathfrak{E}e^{2\underline{C}_Nt}$ and integrate in $p$ yields 
$$\int_{\mathbb{R}^3} \partial_t  G^2\mathfrak{E}e^{2\underline{C}_Nt}dp \ + \ 2\underline{C}_N\int_{\mathbb{R}^3} G^2\mathfrak{E}e^{2\underline{C}_Nt} dp \ \le \ 2M_{|f_0|}\overline{C}_N \int_{\mathbb{R}^3} G e^{2\underline{C}_Nt} dp,$$
which immediately leads to
$$ \partial_t \left(\int_{\mathbb{R}^3} G^2\mathfrak{E}e^{2\underline{C}_Nt}dp\right) \ \le \ 2M_{|f_0|}\overline{C}_N \int_{\mathbb{R}^3} G e^{2\underline{C}_Nt} dp.$$
By H\"older inequality, the right hand side of the above is bounded by
$$\ 2M_{|f_0|}\overline{C}_N \int_{\mathbb{R}^3} G e^{2\underline{C}_Nt} dp \ \le \ \ 2M_{|f_0|}\overline{C}_N \left(\int_{\mathbb{R}^3} G^2 e^{2\underline{C}_Nt} dp \right)^{1/2} e^{\underline{C}_Nt},$$
which yields the following differential inequality
$$ \partial_t \left(\int_{\mathbb{R}^3} G^2\mathfrak{E}e^{2\underline{C}_Nt}dp\right) \ \le \ 2M_{|f_0|}\overline{C}_N \left(\int_{\mathbb{R}^3} G^2 e^{2\underline{C}_Nt} dp \right)^{1/2} e^{\underline{C}_Nt}.$$
Solving the above inequality, we conclude that 
 that the integral $\int_{\mathbb{R}^3}G^2\mathfrak{E}^{-1}dp$ is bounded uniformly in time by some constant $C>0$. As a consequence
\begin{equation}\label{PropNoE:Kinetic:E35}
\begin{aligned}
B_{21}[f] \ \le & \   C,
\end{aligned}
\end{equation}
where $C$ is some universal constant. \\
Combining the Inequalities \eqref{PropNoE:Kinetic:E31}-\eqref{PropNoE:Kinetic:E35}, we find that $B[f]$ is bounded by a universal constant $C$, which, together with \eqref{PropNoE:Kinetic:E30} implies $$A[f]\ \ge \ \left[\int_{\mathbb{R}^3\times\mathbb{R}^3}\left|\hat{f}\right|^2\mathfrak{E}^{-1}dpd\zeta \right]^2.$$
As a result, Inequality \eqref{PropNoE:Kinetic:E29} leads to 
\begin{equation}\label{PropNoE:Kinetic:E36}
\begin{aligned}
\partial_t\int_{\mathbb{R}^3}\mathcal{E}[f](t,\zeta)d\zeta \     \le &  \  -\left(\int_{\mathbb{R}^3}\mathcal{E}[f](t,\zeta)d\zeta \right)^2.\end{aligned}
\end{equation}
Therefore 
\begin{equation}\label{PropNoE:Kinetic:E37}
\begin{aligned}
\int_{\mathbb{R}^3}\mathcal{E}(t,\zeta)d\zeta \     \le &  \  \frac{C}{1+t},\end{aligned}
\end{equation}  
where $C(\mathcal{E}(0,\cdot))$ is some universal constant depending on $\mathcal{E}(0,\cdot)$, which, by \eqref{PropNoE:Kinetic:E6}, implies
\begin{equation}\label{PropNoE:Kinetic:E38}
\int_{\mathbb{R}^3\times\mathbb{R}^3}|\hat{f}|^2\mathfrak{E}^{-1}dpd\zeta\ \le\ \frac{C(\|f_0\|)_{\mathcal{L}}}{1+t}.
\end{equation}
 The second  decay estimate \eqref{PropNoE:Kinetic:2} can be proved by H\"older inequality as follows:
\begin{eqnarray*}
&& \int_{\mathbb{R}^3}\left|\int_{\mathbb{R}^3}f(t,r,p)dp - M_f\right|^2 dr\\
&= &\int_{\mathbb{R}^3}\left|\int_{\mathbb{R}^3}[f(t,r,p) - M_f\mathfrak{E}(p)]dp\right|^2 dr\\
& \le &\int_{\mathbb{R}^3}\int_{\mathbb{R}^3}[f(t,r,p) - M_f\mathfrak{E}(p)]^2\mathfrak{E}^{-1}(p)dp\int_{\mathbb{R}^3}\mathfrak{E}(p)dp dr\\
& \le &\int_{\mathbb{R}^3}\int_{\mathbb{R}^3}[f(t,r,p) - M_f\mathfrak{E}(p)]^2\mathfrak{E}^{-1}(p)dpdr\\
&\le &\frac{C\big(\|f_0\|_{\mathcal{L}}\big)}{1+t}. 
\end{eqnarray*}
\end{proof}

\subsection{The decay rates when $L=L_2$}
Let us start by the following weighted Poincar\'e inequality, whose proof can be found in the Appendix and is inspired by a remark of P.-L. Lions \cite{EscobedoKavian:1987:VPR}, to prove the classical Poincar\'e inequality with inverse Gaussian weight
\begin{equation}\label{Poincare}
\int_{\mathbb{R}^3}e^{|p|^2}|\nabla\varphi(p)|^2 dp \ \ge \ C_{PC}\int_{\mathbb{R}^3}e^{|p|^2}|\varphi(p)|^2 dp,
\end{equation}
for some universal constant $C_{PC}$. We would like to thank E. Zuazua for showing us the remark. 
\begin{lemma}\label{Lemma:Poincare}
We have the following Poincar\'e inequality with inverse Bose-Einstein Distribution  weight, for all function $\varphi$ such that all the integrals below are well defined:
\begin{equation}\label{Poincare}
\int_{\mathbb{R}^3}\mathfrak{E}^{-1}(p)|\nabla\varphi(p)|^2 dp \ \ge \ \frac{1}{4}\int_{\mathbb{R}^3}\mathfrak{E}^{-1}(p)|\varphi(p)|^2 dp.
\end{equation}
\end{lemma}
\begin{proposition}\label{Pro:Kinetic}
Suppose that $f_0$ be a positive function in $L^1(\mathbb{R}^3\times\mathbb{R}^3)\cap \mathcal{L}$ and $N_c(t,r)$ is bounded from above by $\overline{C}_N>0$ and from below by $\underline{C}_N>0$. Under the assumption that  $L=L_2$, Equation \eqref{System1}-\eqref{System1b} has a unique positive solution $f$, which decays exponentially in time towards the equilibrium $f_\infty$  in the following sense: there exist  $\mathcal{C}_1>0$ depending only on $\mathfrak{E}$, such that
\begin{equation}
\label{Prop:Kinetic:1}
\|f(t)-f_\infty\|_{\mathcal{L}}\le \|f_0-f_\infty\|_{\mathcal{L}}e^{-\mathcal{C}_1 t}. 
\end{equation}
Moreover, there also exists $\mathcal{C}_2>0$ depending only on $\mathfrak{E}$ such that
\begin{equation}
\label{Prop:Kinetic:2}
\|\rho[f](t)-M_{f_0}\|_{L^2_r(\mathbb{R}^3)}\le {\|f_0-f_\infty\|_{\mathcal{L}}}e^{-\mathcal{C}_2 t}.
\end{equation}
If $\|\nabla N_c\|_{L^\infty(\mathbb{R}^3)}$ is also bounded by a constant $C_N^*$, there exists $\mathcal{C}_3>0,\mathcal{C}_4>0$ depending on $\|f_0-f_\infty\|_{\mathcal{L}},C_N^*$, such that
 \begin{equation}
\label{Prop:Kinetic:3}
\|f(t)-f_\infty\|_{\mathcal{L}} \ + \ \|\nabla f(t)\|_{\mathcal{L}}\le 3\left(t \|\nabla N_c\|_{L^\infty_r}\|f_0\|_{\mathcal{L}} \ + \ \|\nabla f_0\|_{\mathcal{L}} \ +\ \|f_0\|_{\mathcal{L}}\right)e^{-\mathcal{C}_*t/2}. 
\end{equation}
\end{proposition}
\begin{proof}
Similar as in Proposition \ref{PropNoE:Kinetic}, the existence and uniqueness result of the equation \eqref{System1} is classical. 
\\ We can  assume without loss of generality that $f_\infty=0$. 
Using $f\mathfrak{E}^{-1}$ as a test function in \eqref{System1} yields \begin{equation}\label{Prop:Kinetic:E2}
\frac{d}{dt} \int_{\mathbb{R}^3\times\mathbb{R}^3}|f|^2\mathfrak{E}^{-1}dpdr \ = 
 - \int_{\mathbb{R}^3\times\mathbb{R}^3}2 N_c|\nabla f|^2 \mathfrak{E}^{-1}dpdr.
\end{equation}
It follows directly from Lemma \ref{Lemma:Poincare} that 
\begin{equation}\label{Prop:Kinetic:E3}
 \  2\int_{\mathbb{R}^3\times\mathbb{R}^3} N_c|\nabla f|^2 \mathfrak{E}^{-1}dpdr  \ \ge \ \frac{\underline{C}_N}{2}\int_{\mathbb{R}^3\times\mathbb{R}^3}|f|^2\mathfrak{E}^{-1}dpdr
\end{equation}
 Putting together the two Inequalities yields  
\begin{equation}\label{Prop:Kinetic:E3a}
\|f\|_{\mathcal{L}}\ \le \ \|f_0\|_{\mathcal{L}} e^{-\underline{C}_Nt/4}.
\end{equation}
The second  decay estimate \eqref{Prop:Kinetic:2} can be proved by H\"older inequality as for \eqref{PropNoE:Kinetic:2} of Proposition \ref{PropNoE:Kinetic}. 
\\ 
We now prove the third decay estimate \eqref{Prop:Kinetic:3}. Defining $g_i=\partial_{r_i} f$, where $r_i$ is one of the component of the space variable $r=(r_1,r_2,r_3)\in\mathbb{R}^3$, and taking $g_i\mathfrak{E}^{-1}$ as a test function, we find
\begin{equation}\begin{aligned}\label{Prop:Kinetic:E4a}
\frac{d}{dt}\int_{\mathbb{R}^3\times \mathbb{R}^3}|g_i|^2\mathfrak{E}^{-1} drdp   \ \le &\ -\frac{\underline{C}_N}{2}\int_{\mathbb{R}^3\times\mathbb{R}^3}|f|^2\mathfrak{E}^{-1}dpdr   \ +\ 2\int_{\mathbb{R}^3\times \mathbb{R}^3}\partial_i N_c fg_i \mathfrak{E}^{-1} drdp.
\end{aligned}
\end{equation}
Now, we can estimate the second term on the right hand side of \eqref{Prop:Kinetic:E4a} as follows
\begin{equation*}\begin{aligned}
2\int_{\mathbb{R}^3\times \mathbb{R}^3}\partial_i N_c  fg_i \mathfrak{E}^{-1} drdp \ \le & \ 2\|\nabla N_c\|_{L^\infty_r}\|f\|_{\mathcal{L}}\|g_i\|_{\mathcal{L}},
\end{aligned}
\end{equation*}
which, together with \eqref{Prop:Kinetic:E4a}  leads to
\begin{equation}\label{Prop:Kinetic:E5}\begin{aligned}
\frac{d}{dt} \int_{\mathbb{R}^3\times \mathbb{R}^3}|g_i|^2\mathfrak{E}^{-1} drdp  \ +\ \frac{\underline{C}_N}{2}\int_{\mathbb{R}^3\times\mathbb{R}^3}|f|^2\mathfrak{E}^{-1}dpdr \ \le &\ \ 2\|\nabla N_c\|_{L^\infty_r}\|f\|_{\mathcal{L}}\|g_i\|_{\mathcal{L}}.
\end{aligned}
\end{equation}
Plugging the decay estimate \eqref{Prop:Kinetic:E3a} into \eqref{Prop:Kinetic:E5} implies
\begin{equation}\label{Prop:Kinetic:E5a}\begin{aligned}
\frac{d}{dt} \|g_i\|_{\mathcal{L}}^2  \ +\ \frac{\underline{C}_N}{2}\|g_i\|_{\mathcal{L}}^2 \ \le &\ 2\|\nabla N_c\|_{L^\infty_r}\|f_0\|_{\mathcal{L}} e^{-\underline{C}_N t/4}\|g_i\|_{\mathcal{L}},
\end{aligned}
\end{equation}
which yields
\begin{equation}\label{Prop:Kinetic:E6}\begin{aligned}
\|g_i\|_{\mathcal{L}} \ \le \ \left(t \|\nabla N_c\|_{L^\infty_r}\|f_0\|_{\mathcal{L}} \ + \ \|g_i(0)\|_{\mathcal{L}}\right)e^{-\underline{C}_N t/4}.
\end{aligned}
\end{equation}
As a consequence, \eqref{Prop:Kinetic:3} follows.
\end{proof}
  \section{The defocusing cubic nonlinear Schr\"odinger equation}\label{Sec:NLS}
In this section, we consider the scattering theory for the following defocusing cubic nonlinear Schr\"odinger equation
\begin{eqnarray}\label{Sec:NLS:E1}
i  \frac{\partial \Psi({r},t)}{\partial t} \ &=& \ \Big(-{ \Delta_{{r}}} \ + \ |\Psi({r},t)|^2 + W(t,r) \Big)\Psi({r},t),\\\label{Sec:NLS:E2}
\Psi(0,r)&=&\Psi_0(r), \forall r\in\mathbb{R}^3,
\end{eqnarray}
where 
\begin{equation}
\label{Sec:NLS:E3a}
\|V(t,\cdot)\|_{H^1_r(\mathbb{R}^3)} \ = \ \|W(t,\cdot)+1\|_{H^1_r(\mathbb{R}^3)} \ \le  \ \mathfrak{C}_1 e^{-\mathfrak{C}_2 t}, 
\end{equation}
and 
\begin{equation}
\label{Sec:NLS:E3b}
\|V(t,\cdot)+1\|_{L^1_r(\mathbb{R}^3)}  \ \le  \ \mathfrak{C}_3, 
\end{equation}
with some  positive constants $\mathfrak{C}_1, \mathfrak{C}_2, \mathfrak{C}_3$.
\\ From \eqref{Sec:NLS:E3a} and \eqref{Sec:NLS:E3b}, we deduce that
\begin{equation}
\label{Sec:NLS:E3c}
\|V(t,\cdot)\|_{L^{3/2}_r(\mathbb{R}^3)} \ = \ \|W(t,\cdot)+1\|_{L^{3/2}_r(\mathbb{R}^3)} \ \le  \ \mathfrak{C}_4 e^{-\mathfrak{C}_5 t}, 
\end{equation}
with some  positive constants $\mathfrak{C}_4, \mathfrak{C}_5$.
\\ We denote the Fourier transform on $\mathbb{R}^3$ by
\begin{equation}\label{Prepro:NLS:E1}
\begin{aligned}
\mathcal{F}\varphi \ = \ \hat{\varphi}(\zeta) \ := &
\ \int_{\mathbb{R}^3}\varphi(r)e^{-ir\zeta}dr,\\
\mathcal{F}_r^\zeta[f(r,r')] \ = & \ \left(\mathcal{F}_r^\zeta\right)(\zeta,r') \ := \ \int_{\mathbb{R}^3}f(r,r')e^{-ir\zeta}dr,
\end{aligned}
\end{equation}
as well as the Fourier multiplier 
\begin{equation}\label{Prepro:NLS:E2}
\begin{aligned}
\varphi(-i\nabla)f \ := &\ \mathcal{F}^{-1}[\varphi(\zeta)\hat{f}(\zeta)],\\
\varphi(-i\nabla)_rf(r,r') \ := &\ (\mathcal{F}_r^\zeta)^{-1}[\varphi(\zeta)\mathcal{F}^\zeta_r[f(r,r')]].
\end{aligned}
\end{equation}
Next, we define the standard Littlewood-Paley decomposition. Let $\chi$ be a fixed cut-off function $\chi\in C_0^\infty(\mathbb{R})$ satisfying $\chi(r)=1$ for $|r|\le 1$ and $\chi(r)=0$ for $|r|\geq 2$. Define for each $k\in 2^\mathbb{Z}$ the function
\begin{equation}\label{Prepro:NLS:E3}
\chi^k(r) \ := \ \chi(|r|/k) \ - \ \chi(2|r|/k),
\end{equation}
such that $\chi^k\in C_0^\infty(\mathbb{R}^3)$ and
\begin{equation}\label{Prepro:NLS:E4}
\mathrm{supp} \chi^k \ \subset \ \{k/2 \ < \ |r| \ < \ 2k\},\ \ \ \ \ \ \ \ \ \sum_{k\in 2^{\mathbb{Z}}} \chi^k(r) \ = \ 1\ \ \ (r\ne 0). 
\end{equation}
The Littlewood-Paley decomposition is then defined as follows 
\begin{equation}\label{Prepro:NLS:E5}
f \ = \ \sum_{k\in 2^{\mathbb{Z}}} \chi^k(\nabla) f,
\end{equation}
which leads to the following decomposition into lower and higher frequencies
\begin{equation}\label{Prepro:NLS:E6}
f_{<k} \ := \ \sum_{j<k} \chi^j(\nabla) f,\ \ \ \  f_{\ge k} \ := \ \sum_{j\ge k} \chi^j(\nabla) f.
\end{equation}
For any function $B(\zeta_1,\dots,\zeta_N)$ on $(\mathbb{R}^3)^N$, we define the $N$-multilinear operator $B[f_1,\dots,f_N]$ 
\begin{equation}\label{Prepro:NLS:E6a}
\mathcal{F}^\zeta_rB[f_1,\dots,f_N] \ : = \ \int_{\zeta=\zeta_1+\dots +\zeta_N}B(\zeta_1,\dots,\zeta_N)\hat{f}_1(\zeta_1)\dots\hat{f}_N(\zeta_N)d\zeta_1\dots\zeta_N.
\end{equation}
The above operator is known as a multilinear Fourier multiplier with symbol $B$.\\
Let us also recall inequality $(2.20)$ in \cite{GustafsonNakanishiTsai:2009:STF}: 
\\ For $k\in\mathbb{N}$ and $$\frac{1}{p_0}=\frac{1}{p_1}+\frac{1}{p_2} ,  p_0,p_1,p_2\in (1,\infty),$$ the following inequality holds true
\begin{equation}\label{Pro:NLS:E15a}
\sup_{a\in[0,1]}\left\|\frac{\langle\zeta_1\rangle^{2k(1-a)}\langle\zeta_2\rangle^{2ka}}{\langle(\zeta_1,\zeta_2)\rangle^{2k}}[f,g]\right\|_{L^{p_0}_r(\mathbb{R}^3)} \ \lesssim\ \|f\|_{L_r^{p_1}(\mathbb{R}^3)}\|g\|_{L_r^{p_2}(\mathbb{R}^3)}.
\end{equation}
\\ Define 
\begin{equation}\label{Sec:NLS:E4}
u \ = \ \Psi \  - \ 1 \ = u_1 \ - \ i u_2,
\end{equation}
we obtain the following system of equations whose solution is $(u_1,u_2)$
\begin{equation}\label{Sec:NLS:E5}
\begin{aligned}
\dot{u}_1 \ = & \ -\Delta u_2 \ + \ 2u_1u_2 \ + \ |u|^2u_2 \ + V u_2, \\
\dot{u}_2 \ = & \ -(2-\Delta) u_1 \ - \ 3u_1^2 \ - \ u_2^2 \ - \ |u|^2u_1 \ - \ V (u_1+1).
\end{aligned}
\end{equation}
We define 
\begin{equation}\label{Sec:NLS:E6}
v \ = \ u_1 \ + \  i U u_2, \ \ U \ = \ \sqrt{-\Delta (\ 2 \ - \ \Delta)^{-1}},
\end{equation}
and obtain the following equation for $v$, 
\begin{equation}\label{Sec:NLS:E7}
i\partial_t v \ - \ Hv \ = \ U(3u_1^2\ + \ u_2^2 \ + \ |u|^2u_1) \ + \ i(2u_1u_2 \ + \ |u|^2u_2),
\end{equation}
where
\begin{equation}\label{Sec:NLS:E8}
H \ = \ \sqrt{-\Delta (\ 2 \ - \ \Delta)}.
\end{equation}
For any number or vector $\zeta$, let us define
\begin{equation}\label{Sec:NLS:E9}
\langle\zeta\rangle \ := \ \sqrt{2 \ +  \ |\zeta|^2},\ \ \ U(\zeta) \ := \ \frac{|\zeta|}{\zeta}, \ \ \  H(\zeta) \ := \ |\zeta|\langle \zeta\rangle, \ \ \tilde{\zeta}\ := \ \frac{\zeta}{|\zeta|},
\end{equation}
which will appear normally in Fourier spaces, and the operators $U$ and $H$ in \eqref{Sec:NLS:E9} are the same with the ones defined in \eqref{Sec:NLS:E8} and \eqref{Sec:NLS:E6}. 
\begin{proposition}\label{Pro:NLS}
For $\delta>0$ small enough, such that for $\Phi_0\in H^1(\mathbb{R}^3)$ satisfying
\begin{equation}\label{Pro:NLS:1}
\int_{\mathbb{R}^3}\langle r\rangle^2 \left(|\mathrm{Re}\Phi_0(r)|^2+|\nabla \Phi_0(r)|^2\right)dr<\delta^2,
\end{equation} 
the Equation \eqref{Sec:NLS:E1} has a unique global solution $\Psi_\delta=1+u$ such that  for $v:=\mathrm{Re} u +iU\mathrm{Im}u$, we have $e^{itH}v\in C(\mathbb{R}; \langle r\rangle^{-1}H^1_r(\mathbb{R}^3))$. Moreover,
\begin{equation}\label{Pro:NLS:2}
\|v(t)-e^{-itH}v_+\|_{H^1_r} \ \le \frac{M_1(\mathfrak{C}_1,\delta)}{\sqrt{t+1}}, \  \|\langle r\rangle \left[e^{itH}v(t)-v_+\right]\|_{H^1_r} \le M_0(t,\mathfrak{C}_1,\delta)\to 0,
\end{equation}
as $t\to 0$, for some $v_+\in \langle r\rangle^{-1} H^1_r(\mathbb{R}^3)$.\\
Define $u=u_{1}+ iu_{2}$, we also have
\begin{equation}\label{Pro:NLS:3}
\|u_{1}(t)\|_{L^\infty_r}\leq \frac{M_2(\mathfrak{C}_1,\delta)}{{t+1}},~~~~\|u_{2}(t)\|_{L^\infty_r}\leq \frac{M_3(\mathfrak{C}_1,\delta)}{{(t+1)}^{9/10}}.\end{equation}
Moreover, for fixed $\mathfrak{C}_1$, the three functions $M_0,M_1, M_2, M_3$ are  decreasing  in $\delta$ and tend to $0$ as $\delta$ and $\mathfrak{C}_1$ tend to $0$.
\end{proposition}
\begin{proof}
\\ Similar as in \cite{GustafsonNakanishiTsai:2009:STF}, we also define
\begin{equation}\label{Sec:NLS:E10}
Z\ := \ v \ + \ b(u) \ := \ v \ - \ \langle (\zeta_1,\zeta_2) \rangle^{-2}[u_1,u_1] \ + \ \langle (\zeta_1,\zeta_2) \rangle^{-2}[u_2,u_2],
\end{equation}
and obtain the following equation for $Z$ by the {\it normal form transformation}
\begin{equation}\label{Sec:NLS:E11}
i\dot{Z} \ - \ HZ \ = \ \mathcal{N}_Z(v) \ + \ \mathcal{M}(v),
\end{equation}
in which 
\begin{equation}\label{Sec:NLS:E11a}
\mathcal{M}(v) \ = \  -2\langle(\zeta_1,\zeta_2)\rangle^{-2}[u_1,Vu_2] \ - 2 \langle (\zeta_1,\zeta_2)\rangle^{-2} [u_2,V(u_1+1)],
\end{equation}
and the nonlinear term $\mathcal{N}_Z(v)$ is of the following form
\begin{equation}\label{Sec:NLS:E12}
\begin{aligned}
\mathcal{N}_Z(v) \ :=& \ B_1[v_1,v_1] \ + \  B_2[v_2,v_2] \ + \ C_1[v_1,v_1,v_1] \ + \ C_2[v_2,v_2,v_1]\\
& \ + \ i C_3[v_1,v_1,v_2] \ + \ i C_4[v_2,v_2,v_2]  \ + \ iQ_1[u],
\end{aligned}
\end{equation}
with the following definitions for $B_1$ and $B_2$ in the Fourier space 
\begin{equation}\label{Sec:NLS:E13}
\begin{aligned}
B_1(\zeta_1,\zeta_2)  \ = &\ \frac{-2U(\zeta_1 \ + \ \zeta_2) (4 \ + \ 4|\zeta_1|^2 \ + \ 4|\zeta_2|^2 \ - \ \zeta_1\zeta_2)}{2 \ +  \ |\zeta_1|^2 \ + \ |\zeta_2|^2},
\\
B_2(\zeta_1,\zeta_2)  \ = &\ \frac{-2U(\zeta_1 \ + \ \zeta_2)\langle\zeta_1\rangle\langle\zeta_2\rangle\tilde{\zeta_1}\tilde{\zeta_2}}{2 \ +  \ |\zeta_1|^2 \ + \ |\zeta_2|^2},
\end{aligned}
\end{equation}
and the cubic multipliers are defined in the Fourier space as follows
\begin{equation}\label{Sec:NLS:E14}
\begin{aligned}
C_1(\zeta_1 \ + \ \zeta_2,\zeta_1,\zeta_2)  \ =& \ U(\zeta_1 \ + \ \zeta_2), \ \ C_2(\zeta_1 \ + \ \zeta_2,\zeta_1,\zeta_2) \ = \ U(\zeta_1 \ + \ \zeta_2)U(\zeta_1)^{-1}U(\zeta_2)^{-1},\\
C_3(\zeta_1,\zeta_2,\zeta_3) \ =& \ U(\zeta_3)^{-1}\left(1 \ - \  \frac{4}{2 \ + \ |\zeta_1|^2 \ + \ |\zeta_2+\zeta_3|^2} -  \frac{6}{2 \ + \ |\zeta_1+\zeta_2|^2  + |\zeta_3|^2}\right),\\
C_4(\zeta_1,\zeta_2,\zeta_3) \ =& \ U(\zeta_3)^{-1}\left(1 \  - \  \frac{2}{2 \ + \ |\zeta_1+\zeta_2|^2 \ + \ |\zeta_3|^2}\right).\
\end{aligned}
\end{equation} 
Moreover, $Q_1$ is of the following form
\begin{equation}\label{Sec:NLS:E15}
Q_1(u) \ = \ -2\langle(\zeta_1,\zeta_2)\rangle^{-2}[u_1,|u|^2u_2] \ - 2 \langle (\zeta_1,\zeta_2)\rangle^{-2} [u_2,|u|^2(u_1+1)].
\end{equation}
For any complex-valued function $f$, set 
\begin{equation}\label{Sec:NLS:E16}
Jf \ = \ e^{-itH} r e^{itH}f.
\end{equation}
Now, our function spaces can be set up as follows
\begin{equation}\label{Sec:NLS:E17}
\begin{aligned}
\|Z(t)\|_{X(t)} \ :=& \ \|Z(t)\|_{H_r^1} \ + \ \|JZ(t)\|_{H^1_r}, \ \ \|Z\|_X \ := \ \mathrm{sup}_t\|Z(t)\|_{X(t)},\\
\|Z\|_S \ :=& \ \|Z\|_{L^\infty_tH^1_r}\ + \ \|U^{-1/6}Z\|_{L^2_t H^{1,6}_r}.
\end{aligned}
\end{equation}
Fix a time $T$ large enough, by Duhamel formula, applied to \eqref{Sec:NLS:E11}, we find
\begin{equation}\label{Pro:NLS:E1}
\begin{aligned}
Z(t) \ = & \ e^{-iHt}Z(T) \ + \ \int_T^t e^{-iH(t-s)}\mathcal{N}_Z(s) ds \ + \ \int_T^t e^{-iH(t-s)}\mathcal{M}(s) ds. 
\end{aligned}
\end{equation}
By Inequalities $(5.3)$ and $(8.5)$ of \cite{GustafsonNakanishiTsai:2009:STF}, we have that
\begin{equation}\label{Pro:NLS:EstimateI0}
\begin{aligned}
\left\|\int_T^t e^{-i(t-s)H}\mathcal{N}_Z(s)ds\right\|_{X(T,\infty)}\  \lesssim & \ \langle T\rangle^{-\epsilon},\\
\left\|\int_T^t e^{-i(t-s)H}\mathcal{N}_Z(s)ds\right\|_{S(T,\infty)}\  \lesssim & \ T^{-1/2}\left(\|v\|_{X\cap S}^2 \ + \ \|v\|_{X\cap S}^4 \right),\\
\end{aligned}
\end{equation}
for some small $\epsilon>0$. 
\\ Now, we will estimate the left-over in the norm \eqref{Sec:NLS:E17}
\begin{equation}\label{Pro:NLS:E3}
I \ := \ \int_T^t e^{-iH(t-s)}\mathcal{M}(s) ds. 
\end{equation}
Let us define 
\begin{equation}
 I_1:=\|J\mathcal{M}\|_{L^\infty_t H^1_r},
\end{equation} which can be bounded as
\begin{equation}\label{Pro:NLS:E14}
\begin{aligned}
I_1 \  \lesssim &\  \left\|\int_T^t e^{-iH(t-s)}\left(r-s\nabla H(\zeta)\right)\mathcal{M}(s) ds\right\|_{L^\infty_tH^1_r}.
\end{aligned}
\end{equation}
By Strichartz inequality, the above inequality can be estimated as 
\begin{equation}\label{Pro:NLS:E15}
\begin{aligned}
I_1 \  \lesssim &\  \|r\mathcal{M}\|_{L^2_tH^{1,6/5}_r}  + \ \|t\mathcal{M}\|_{L^2_tH^{2,6/5}_r}.
\end{aligned}
\end{equation}
Using  \eqref{Pro:NLS:E15a} for $k=1$, $p_0=6/5$, $p_1=6$, $p_2=3/2$, we find
\begin{equation}\label{Pro:NLS:E15b}
\begin{aligned}
 \|r\mathcal{M}\|_{H^{1,6/5}_r} \  \lesssim &\ \|r u\|_{L^6_r }\|Vu\|_{L^{3/2}_r}\ + \ \|r u\|_{L^6_r}\|V\|_{L^{3/2}_r},
\end{aligned}
\end{equation}
which, by H\"older inequality, can be estimated as
\begin{equation}\label{Pro:NLS:E16}
\begin{aligned}
 \|r\mathcal{M}\|_{L^2_tH^{1,6/5}_r} \  \lesssim &\ \|r u\|_{L^\infty_t L^6_r }\|u\|_{L^\infty_tL^\infty_r}\|V\|_{L^2_t L^{3/2}_r}\ + \ \|r u\|_{L^\infty_t L^6_r }\|V\|_{L^2_t L^{3/2}_r}.
\end{aligned}
\end{equation}
Recall Inequality $(9.9)$ from \cite{GustafsonNakanishiTsai:2009:STF},
\begin{equation}\label{Pro:NLS:E17}
\|ru\|_{L^6_r} \ \le  \ \|v(t)\|_{X(t)},
\end{equation} 
and Inequality $(5.8)$ in \cite{GustafsonNakanishiTsai:2009:STF}
\begin{equation}\label{Pro:NLS:E9}
\begin{aligned}
\| |\nabla|^{-2+5\theta/3} v_{<1}(t)\|_{L^6_r} \ \lesssim & \ \mathrm{min}\left(1,t^{-\theta}\right)\|v(t)\|_{X(t)},\\
\| |\nabla|^{\theta} v_{\ge 1}(t)\|_{L^6_r} \ \lesssim & \ \mathrm{min}\left(t^{-\theta},t^{-1}\right)\|v(t)\|_{X(t)}, \ \ \ \forall\theta\in[0,1].
\end{aligned}
\end{equation}
Moreover, we also have
 \begin{equation}\label{Pro:NLS:E6}
 \begin{aligned}
\|V\|_{L^\infty_tL^2_r} \ \lesssim \ & \langle T\rangle^{-n},\\
\|\nabla V\|_{L^\infty_tL^2_r} \ \lesssim \ & \langle T\rangle^{-n}, \ \ \ \forall n>0.
\end{aligned}
\end{equation}
By using the boundedness of $u$, Inequalities \eqref{Pro:NLS:E9}, the decays \eqref{Sec:NLS:E3c}, \eqref{Pro:NLS:E6} of $V$ and \eqref{Pro:NLS:E17}, we deduce from \eqref{Pro:NLS:E16} that
\begin{equation}\label{Pro:NLS:E18}
\begin{aligned}
 \|r\mathcal{M}\|_{L^2_tH^{1,6/5}_r} \  \lesssim &\  \left(\|v\|_{X}+1\right)\langle T \rangle^{-n}, \ \ \forall n>0.
\end{aligned}
\end{equation}
Using  \eqref{Pro:NLS:E15a} for $k=1$, $p_0=6/5$, $p_1=6$, $p_2=3/2$, we find
\begin{equation}\label{Pro:NLS:E15b}
\begin{aligned}
 \|t\mathcal{M}\|_{H^{2,6/5}_r} \  \lesssim &\ \|t u\|_{L^6_r }\|Vu\|_{L^{3/2}_r}\ + \ \|t u\|_{L^6_r}\|V\|_{L^{3/2}_r},
\end{aligned}
\end{equation}
which, again by H\"older inequality, can be bounded as
 \begin{equation}\label{Pro:NLS:E19}
\begin{aligned}
 \|t\mathcal{M}\|_{L^2_tH^{1,6/5}_r} \  \lesssim &\ \|tu\|_{L^\infty_tL^6_r}\|u\|_{L^2_tL^6_r}\|V\|_{L^\infty_tL^2_r} \ +  \ \|tu\|_{L^\infty_tL^6_r}\|V\|_{L^2_tL^{3/2}_r} .
\end{aligned}
\end{equation}
Replacing $\theta=3/5$ and $\theta=0$ into \eqref{Pro:NLS:E9}, we can deduce that
\begin{equation}\label{Pro:NLS:E20}
\begin{aligned}
\| u_{<1}(t)\|_{L^6_r} \ \lesssim \ \| \nabla^{-1} v_{<1}(t)\|_{L^6_r}   \ \lesssim & \ \mathrm{min}\left(1,t^{-3/5}\right)\|v(t)\|_{X(t)},\\
\| u_{\ge 1}(t)\|_{L^6_r} \ \lesssim \ \| v_{\ge 1}(t)\|_{L^6_r} \ \lesssim & \ \mathrm{min}\left(1,t^{-1}\right)\|v(t)\|_{X(t)},
\end{aligned}
\end{equation}
which yields at once, for $t$ large
 \begin{equation}\label{Pro:NLS:E21}
 \begin{aligned}
\|u(t)\|_{L^6_r}\ \lesssim & \ \langle t\rangle^{-3/5}\\
\|u(t)\|_{L^2_tL^6_r}\ \lesssim & \ \langle t\rangle^{-1/10}.
\end{aligned}
\end{equation}
Using \eqref{Pro:NLS:E21} and \eqref{Pro:NLS:E6}, we find that
 \begin{equation}\label{Pro:NLS:E22}
\begin{aligned}
 \|t\mathcal{M}\|_{L^2_tH^{1,6/5}_r}  \
\lesssim & \ \langle T\rangle^{-n},\ \  \forall n>0.
\end{aligned}
\end{equation}
As a consequence, from \eqref{Pro:NLS:E15}, \eqref{Pro:NLS:E18}, \eqref{Pro:NLS:E22},  we deduce
\begin{equation}\label{Pro:NLS:EstimateI2}\begin{aligned}
I_1  
\ \le  & \ \langle T \rangle^{-n}\|v\|_X, \ \ \forall n>0.
\end{aligned}
\end{equation}
Now let us consider
\begin{equation}\label{Pro:NLS:E4a}
\begin{aligned}
I_2 \ := &\ \left\|\int_T^t e^{-iH(t-s)}\mathcal{M}(s) ds\right\|_{L^\infty_tH^1_r}.
\end{aligned}
\end{equation}
As a view of Strichartz estimate, we obtain
\begin{equation}\label{Pro:NLS:E4}
\begin{aligned}
I_2 \ 
 \lesssim &\ \||\nabla V||u|^2\|_{L^\infty_tL^2_r} \ + \  \||V||u|^2\|_{L^\infty_tL^2_r} \ + \||V||\nabla u||u|\|_{L^\infty_tL^2_r}\\\
\ & + \  \|V u\|_{L^\infty_tL^2_r}+ \  \|\nabla V u\|_{L^\infty_tL^2_r}+ \  \|V \nabla u\|_{L^\infty_tL^2_r},
\end{aligned}
\end{equation}
which, by H\"older inequality, can be estimated as 
\begin{equation}\label{Pro:NLS:E5}
\begin{aligned}
\left\|\int_T^t e^{-iH(t-s)}\mathcal{M}(s) ds\right\|_{L^\infty_tH^1_r} \ \lesssim &\ \|\nabla V\|_{L^\infty_tL^2_r} \|u\|_{L^\infty_tL^\infty_r}^2 \ + \  \|V\|_{L^\infty_tL^2_x} \|u\|_{L^\infty_tL^\infty_r}^2\\
\ &\ + \|V\|_{L^\infty_tL^3_r} \|u\|_{L^\infty_tL^\infty_r}\|\nabla u\|_{L^\infty_tL^6_r} \ + \  \|V\|_{L^\infty_tL^2_r} \|u\|_{L^\infty_tL^\infty_r}\\
\ &\ + \  \|\nabla V\|_{L^\infty_tL^2_r} \|u\|_{L^\infty_tL^\infty_x} \ + \ \|V\|_{L^\infty_t L^3_r}\|\nabla u\|_{L^\infty_tL^6_r}.
\end{aligned}
\end{equation}
Using the fact that $\|u\|_{L^\infty_tL^\infty_r}$ is bounded and \eqref{Pro:NLS:E6},
we obtain from \eqref{Pro:NLS:E5} that
\begin{equation}\label{Pro:NLS:E7}
\left\|\int_T^t e^{-iH(t-s)}\mathcal{M}(s) ds\right\|_{L^\infty_tH^1_r} \ \lesssim \   \langle T\rangle^{-n}\|\nabla u\|_{L^\infty_tL^6_r}\ + \   \langle T\rangle^{-n}, \ \ \ \forall n>0.
\end{equation}
We observe that 
\begin{equation}\label{Pro:NLS:E8}
 \begin{aligned}
\|\nabla u(t)\|_{L^6_x} \ = &\ \|\nabla u_{\ge 1}(t)\|_{L^6_x} \ + \ \|\nabla u_{< 1}(t)\|_{L^6_x}\\
\ \approx	 &\ \|\nabla v_{\ge 1}(t)\|_{L^6_r} \ + \ \|v_{< 1}(t)\|_{L^6_r}.
\end{aligned}
\end{equation}
Using \eqref{Pro:NLS:E9} for $\theta=1$, we find
 \begin{equation}\label{Pro:NLS:E11}
\begin{aligned}
\| |\nabla|^{-1/3} v_{<1}(t)\|_{L^6_r}  \ + \ \| |\nabla| v_{\ge 1}(t)\|_{L^6_r}  \ \lesssim & \ \min\left(1,t^{-1}\right)\|v(t)\|_{X(t)} \ + \ t^{-1}\|v(t)\|_{X(t)},
\end{aligned}
\end{equation}
which, together with \eqref{Pro:NLS:E8}, leads to
\begin{equation}\label{Pro:NLS:E12}
\|\nabla u(t)\|_{L^6_r} \ \lesssim \mathrm{min}\left(1,t^{-1}\right)\|v(t)\|_{X(t)} \ + \ t^{-1}\|v(t)\|_{X(t)}.
\end{equation}
With Inequality \eqref{Pro:NLS:E12}, we can bound \eqref{Pro:NLS:E7} as 
\begin{equation}\label{Pro:NLS:EstimateI1}
I_2\ \lesssim \   \langle T\rangle^{-n}\|Z\|_{X}\ + \   \langle T\rangle^{-n}, \ \ \forall n>0.
\end{equation}
Finally, we  define
\begin{equation}\begin{aligned}\label{Pro:NLS:E26}
I_3 \ := & \ \left\|\int_T^t e^{-iH(t-s)}U^{-1/6}\mathcal{M}ds\right\|_{L^2_tH^{1,6}_r} 
\end{aligned}
\end{equation}
which can be bounded, by Strichartz estimate,  as
\begin{equation}\begin{aligned}\label{Pro:NLS:E26a}
I_3 \ \le  & \ \left\|U^{-1/6}\mathcal{M}\right\|_{L^2_tH^{1,6/5}_r}.
\end{aligned}
\end{equation}
Apply Inequality \eqref{Pro:NLS:E15a} for $k=1$, $p_0=6/5$, $p_1=6$, $p_2=3/2$, we find
\begin{equation}\begin{aligned}\label{Pro:NLS:E27}
I_3 \ \le  & \ \|V\|_{L^\infty_tL^{3/2}_r}\left(\|u\|_{L^2_tL^{6}_r}\ + \ \|u^2\|_{L^2_tL^{6}_r}\right)\\
\ \le  & \ \|V\|_{L^\infty_tL^{3/2}_r}\Big(\|u_{\ge 1}\|_{L^2_tL^{6}_r}\ + \ \|u_{< 1}\|_{L^2_tL^{6}_r}+ \ \|u^2_{\ge 1}\|_{L^{6}_r}\  \\
& + \ \|u^2_{< 1}\|_{L^{6}_r}\Big)\\
\ \le  & \ \|V\|_{L^\infty_tL^{3/2}_r}\Big(\|v_{\ge 1}\|_{L^2_tL^{6}_r}\ + \ \||\nabla|^{-1} v_{< 1}\|_{L^2_tL^{6}_r}+ \ \|v^2_{\ge 1}\|_{L^2_tL^{6}_r}\\
 & + \ \||\nabla|^{-1} v^2_{< 1}\|_{L^2_tL^{6}_r}\Big),  
\end{aligned}
\end{equation}
which, due to the fact that $\|v\|_{L^\infty_r}$ is bounded, can be estimated as
\begin{equation}\label{Pro:NLS:E28}\begin{aligned}
I_3  
\ \le  & \ \|V\|_{L^\infty_tL^{3/2}_r}\Big(\|v_{\ge 1}\|_{L^2_tL^{6}_r}\ + \ \||\nabla|^{-1} v_{< 1}\|_{L^2_tL^{6}_r}\Big).
\end{aligned}
\end{equation}
Using \eqref{Pro:NLS:E9} for $\theta=0, 3/5$ and taking into account the decay \eqref{Sec:NLS:E3c}, we obtain
\begin{equation}\label{Pro:NLS:EstimateI3}\begin{aligned}
I_3  
\ \le  & \ \langle T \rangle^{-n}\|v\|_X,\ \  \forall n>0.
\end{aligned}
\end{equation}
Taking into account the estimates \eqref{Pro:NLS:EstimateI0},  \eqref{Pro:NLS:EstimateI2}, \eqref{Pro:NLS:EstimateI1}, and  \eqref{Pro:NLS:EstimateI3}, by a bootstrap argument as in \cite{GustafsonNakanishiTsai:2009:STF}, the conclusion of the proposition then follows.  
\end{proof}
\section{Proof of Theorem \ref{Thm:Main}}
According to Proposition \ref{Pro:Kinetic}, for a given function $\Psi$ satisfying the assumption of the Proposition, there exists a unique global solution $f=\mathcal{F}_1[\Psi]$ to 
\begin{eqnarray}\label{MainProof:E01}
\frac{\partial f}{\partial t}+{p}\cdot\nabla_{{r}} f & = &
N_cL[f] :=  N_c\mathfrak{E}^{-1}(p)\nabla\left(\mathfrak{E}^{-1}(p) \nabla (f-f_\infty)\right), \\\nonumber
&&\ \ \ \ \ \ \ \ \ \ \ \ (t,r,p) \ \in \ \mathbb{R}_+\times\mathbb{R}^3\times\mathbb{R}^3,\\\label{MainProof:E03}
f(0,r,p)\ &=& \ f_0(r,p), (r,p)\in\mathbb{R}^3\times\mathbb{R}^3.
\end{eqnarray}
Moreover
  \begin{equation}
\label{MainProof:E04}
\|\mathcal{F}_1[\Psi](t)-f_\infty\|_{\mathcal{L}} \ + \ \|\nabla \mathcal{F}_1[\Psi](t)\|_{\mathcal{L}}\le \left(t \|\nabla N_c\|_{L^\infty_r}\|f_0\|_{\mathcal{L}} \ + \ \|\nabla f_0\|_{\mathcal{L}} \ +\ \|f_0\|_{\mathcal{L}}\right)e^{-\underline{C}_N t/4}. 
\end{equation}
We recall that $f_\infty$ is defined in \eqref{KineticEquilibrium}. Without loss of generality, we can suppose that $f_0$ is chosen such that $f_\infty=0$. Hence, \eqref{MainProof:E04} is reduced to 
  \begin{equation}
\label{MainProof:E05}
\|\mathcal{F}_1[\Psi](t)\|_{\mathcal{L}} \ + \ \|\nabla \mathcal{F}_1[\Psi](t)\|_{\mathcal{L}}\le \left(t \|\nabla N_c\|_{L^\infty_r}\|f_0\|_{\mathcal{L}} \ + \ \|\nabla f_0\|_{\mathcal{L}} \ +\ \|f_0\|_{\mathcal{L}}\right)e^{-\underline{C}_N t/4}. 
\end{equation}
Due to Proposition \ref{Pro:NLS}, for a given function $f$ satisfying the assumption of the Proposition, there exists a unique global solution $\Psi=\mathcal{F}_2[f]$ to 
\begin{eqnarray}\label{MainProof:E06}
i  \frac{\partial \Psi(t,{r})}{\partial t} \ &=& \ \Big(-{ \Delta_{{r}}} \ + \ |\Psi(t,{r})|^2 -1 \ + \  \rho[f] \ -M_{f_0} \Big)\Psi(t,{r}),  (t,r)\in \mathbb{R}_+\times\mathbb{R}^3,\\\label{MainProof:E07}
\Psi(0,r)&=&\Psi_0(r), \forall r\in\mathbb{R}^3.
\end{eqnarray}
Moreover
\begin{equation}\begin{aligned}\label{MainProof:E08}
\|\mathcal{F}_2[f](t)-1\|_{L^\infty_r}\ \leq & \ \frac{\mathcal{M}\left(\left(t \|\nabla N_c\|_{L^\infty_r}\|f_0\|_{\mathcal{L}} \ + \ \|\nabla f_0\|_{\mathcal{L}} \ +\ \|f_0\|_{\mathcal{L}}\right)e^{-\underline{C}_N t/4},\delta\right)}{{(t+1)}^{9/10}}\\
\|\mathcal{F}_2[f](t)-1\|_{H^1_r}\ \leq & \ \frac{\mathcal{M}\left(\left(t \|\nabla N_c\|_{L^\infty_r}\|f_0\|_{\mathcal{L}} \ + \ \|\nabla f_0\|_{\mathcal{L}} \ +\ \|f_0\|_{\mathcal{L}}\right)e^{-\underline{C}_N t/4},\delta\right)}{{(t+1)}^{1/2}},\end{aligned}
\end{equation}
where $\mathcal{M}=\max\{M_2+M_3,M_1\}$ and $M_1$, $M_2$, $M_3$ are defined in Proposition \ref{Pro:NLS}.
\\ In order to prove that \eqref{System1b}-\eqref{System5} has a unique global solution, it is sufficient to prove that the  function $\mathcal{F}=\mathcal{F}_1 o \mathcal{F}_2$ has a fixed point.
\\ We deduce from \eqref{MainProof:E04} and \eqref{MainProof:E08} that 
\begin{equation}\label{MainProof:E1}
\begin{aligned}
& \ \|\mathcal{F}[\Psi](t)-1\|_{L^\infty_r}\ \le \\
 \le & \ \ \frac{\mathcal{M}\left(\left(t \|\nabla N_c\|_{L^\infty_r}\|f_0\|_{\mathcal{L}} \ + \ \|\nabla f_0\|_{\mathcal{L}} \ +\ \|f_0\|_{\mathcal{L}}\right)e^{-\underline{C}_Nt/4},\delta\right)}{{(t+1)}^{9/10}}.
\end{aligned}
\end{equation}
Let us consider $\nabla_r N_c(t,r)$, which can be written as
\begin{equation}\label{MainProof:E2}
\nabla_r N_c(t,r) \ = \ C^*\int_{\mathbb{R}^3}\frac{r-r'}{|r-r'|}|\Psi(r',t)|^2 e^{-|r-r'|}dr'.
\end{equation}
Since $$\int_{\mathbb{R}^3}\frac{r-r'}{|r-r'|}e^{-|r-r'|}dr'=0,$$
we can rewrite the form of $\nabla_r N_c(t,r)$ as
\begin{equation}\label{MainProof:E3}
\nabla_r N_c(t,r) \ = \ C^*\int_{\mathbb{R}^3}\frac{r-r'}{|r-r'|}\left(|\Psi(r',t)|^2 -1 \right)e^{-|r-r'|}dr',
\end{equation}
whose sup-norm can be bounded as 
\begin{equation}\label{MainProof:E4}\begin{aligned}
\|\nabla_r N_c\|_{L^\infty_r} \ \le & \ {C^*}\||\Psi|^2-1\|_{L^1_r}\left\|e^{-|r'|^2}\right\|_{L^\infty_r} \\ 
\ \le & \ {C^*}\||\Psi|^2-1\|_{L^1_r}
\\ 
\ \le & \ {C^*}\|\Psi-1\|_{L^2_r}\|\Psi+1\|_{L^2_r}.
\end{aligned}
\end{equation}
Combining the above estimate and \eqref{MainProof:E1} yields
\begin{equation}\label{MainProof:E5}
\begin{aligned}
& \ \|\mathcal{F}[\Psi](t)-1\|_{L^\infty_r}\ \le \\
 \le & \ \ \frac{\mathcal{M}\left(\left({2t }\|\Psi-1\|_{L^2_r}\|\Psi+1\|_{L^2_r}\|f_0\|_{\mathcal{L}} \ + \ \|\nabla f_0\|_{\mathcal{L}} \ +\ \|f_0\|_{\mathcal{L}}\right)e^{-\underline{C}_N t/4},\delta\right)}{{(t+1)}^{9/10}}.
\end{aligned}
\end{equation}
Similarly, the following inequality also holds true
\begin{equation}\label{MainProof:E6}
\begin{aligned}
& \ \|\mathcal{F}[\Psi](t)-1\|_{H^1_r}\ \le \\
 \le & \ \ \frac{\mathcal{M}\left(\left({2t }\|\Psi-1\|_{L^2_r}\|\Psi+1\|_{L^2_r}\|f_0\|_{\mathcal{L}} \ + \ \|\nabla f_0\|_{\mathcal{L}} \ +\ \|f_0\|_{\mathcal{L}}\right)e^{-\underline{C}_N t/4},\delta\right)}{{(t+1)}^{9/10}}.
\end{aligned}
\end{equation}
 We deduce from the above that, for $\epsilon$ small enough and for $\|\nabla f_0\|_{\mathcal{L}}, \|f_0\|_{\mathcal{L}},\delta$ small correspondingly, the operator $\mathcal{F}$ maps the ball $B(1,\epsilon)$ of $L^2_r(\mathbb{R}^3)$ into a compact set of $B(1,\epsilon)$. As a consequence, it has a fixed point and the conclusion of the theorem follows by Propositions \ref{Pro:Kinetic} and \ref{Pro:NLS}.  

  ~~ \\{\bf Acknowledgements.} 
This work was partially supported by a grant from the Simons Foundation ($\#$395767 to Avraham Soffer). A. Soffer is partially supported by NSF grant DMS01600749.
M.-B Tran is partially supported by NSF Grant RNMS (Ki-Net) 1107291, ERC Advanced Grant DYCON. M.-B Tran would like to thank  Professor Linda Reichl and Professor Robert Dorfman  for fruiful discussions on the topic.

\section{Appendix: Proof of Lemma \ref{Lemma:Poincare} -  The Poincar\'e Inequality with inverse Bose-Einstein Distribution weight}
Let us first define
\begin{equation}\label{Lemma:Poincare:E1}
\omega \ = \ |p|,
\end{equation}
and 
\begin{equation}\label{Lemma:Poincare:E2}
F \ = \ \mathfrak{E}^{-1/2}\varphi,
\end{equation}
and take the gradient of the above function
\begin{equation*}
\begin{aligned}
\nabla F \ = & \ \nabla\omega \frac{e^\omega}{2(e^\omega-1)^{1/2}}\varphi \ + \ (e^\omega-1)^{1/2}\nabla \varphi,
\end{aligned}
\end{equation*}
which implies
\begin{equation*}
\begin{aligned}
(e^\omega-1)|\nabla \varphi|^2 \ = & \ \left|\nabla F \ -  \ \nabla\omega \frac{e^\omega}{2(e^\omega-1)^{1/2}}\varphi\right|^2 .
\end{aligned}
\end{equation*}
Expanding the right hand side of the above inequality yields
\begin{equation*}
\begin{aligned}
(e^\omega-1)|\nabla \varphi|^2 \ = & \ |\nabla F|^2 \ - \ \nabla F\nabla \omega \frac{\varphi e^\omega}{(e^\omega-1)^{1/2}}\ + \ |\varphi|^2|\nabla \omega|^2\frac{e^{2\omega}}{4(e^\omega-1)}.
\end{aligned}
\end{equation*}
Since $\omega=|p|$ and $e^{2\omega}\geq (e^\omega-1)^2$, the last term on the right hand side of the above inequality can be bounded as 
$$ |\varphi|^2|\nabla \omega|^2\frac{e^{2\omega}}{4(e^\omega-1)}\ \ge \  |\varphi|^2|\nabla \omega|^2\frac{(e^\omega-1)}{4},$$
which yields
\begin{equation}\label{Lemma:Poincare:E3}
\begin{aligned}
(e^\omega-1)|\nabla \varphi|^2 \ \ge  & \ |\nabla F|^2 \ - \ \nabla F\nabla \omega \frac{\varphi e^\omega}{(e^\omega-1)^{1/2}}\ + \ |\varphi|^2|\nabla \omega|^2\frac{(e^\omega-1)}{4}.
\end{aligned}
\end{equation}
Now, let us consider the second term on the right hand side of \eqref{Lemma:Poincare:E3}, that can be rewritten as
\begin{equation*}
\begin{aligned}
- \nabla F\nabla \omega \frac{\varphi e^\omega}{(e^\omega-1)^{1/2}} \ = &\ - \nabla F\nabla \omega \frac{F e^\omega}{e^\omega-1}\\
= &\ - \frac{1}{2}\nabla F^2\nabla \omega \frac{e^\omega}{e^\omega-1},
\end{aligned}
\end{equation*}
which, in combination with \eqref{Lemma:Poincare:E3}, yields
\begin{equation*}
\begin{aligned}
(e^\omega-1)|\nabla \varphi|^2 \ \ge  & \ |\nabla F|^2 \ - \ \frac{1}{2}\nabla F^2\nabla \omega \frac{e^\omega}{e^\omega-1}\ + \ |\varphi|^2|\nabla \omega|^2\frac{(e^\omega-1)}{4}.
\end{aligned}
\end{equation*}
Integrating both sides of the above inequality with respect to $p$ leads to the following inequality
\begin{equation}\label{Lemma:Poincare:E4}
\begin{aligned}
\int_{\mathbb{R}^3}(e^\omega-1)|\nabla \varphi|^2dp   \ge  &  \int_{\mathbb{R}^3}|\nabla F|^2dp \ - \ \frac{1}{2}\int_{\mathbb{R}^3}\nabla F^2\nabla \omega \frac{e^\omega}{e^\omega-1}dp\ + \ \int_{\mathbb{R}^3}|\varphi|^2|\nabla \omega|^2\frac{(e^\omega-1)}{4}dp\\
 \ge  &  \int_{\mathbb{R}^3}|\nabla F|^2dp  +  \frac{1}{2}\int_{\mathbb{R}^3} F^2\nabla\left(\nabla \omega \frac{e^\omega}{e^\omega-1}\right)dp +  \int_{\mathbb{R}^3}|\varphi|^2|\nabla \omega|^2\frac{(e^\omega-1)}{4}dp,
\end{aligned}
\end{equation}
where the last line follows from an integration by parts on the second term on the right hand side of the inequality. 
\\ Developing the second term on the right hand side of \eqref{Lemma:Poincare:E4}, we find 
\begin{equation}\label{Lemma:Poincare:E5}
\begin{aligned}
&\ \int_{\mathbb{R}^3}(e^\omega-1)|\nabla \varphi|^2dp   \\ 
 \ge  & \ \int_{\mathbb{R}^3}|\nabla F|^2dp \ + \ \frac{1}{2}\int_{\mathbb{R}^3} F^2\left(\Delta \omega \frac{e^\omega}{e^\omega-1} \ - \ \frac{|\nabla \omega|^2}{(e^\omega-1)^2}  \right)dp\ + \ \int_{\mathbb{R}^3}|\varphi|^2|\nabla \omega|^2\frac{(e^\omega-1)}{4}dp.
\end{aligned}
\end{equation}
By noting that $\Delta \omega =\frac{3}{|p|}$ and $|\nabla \omega|=1$, we deduce from \eqref{Lemma:Poincare:E4} that
\begin{equation}\label{Lemma:Poincare:E6}
\begin{aligned}
&\ \int_{\mathbb{R}^3}(e^\omega-1)|\nabla \varphi|^2dp   \\ 
 \ge  & \ \int_{\mathbb{R}^3}|\nabla F|^2dp \ + \ \frac{1}{2}\int_{\mathbb{R}^3} F^2\left(\frac{3}{|p|} \frac{e^{|p|}}{e^{|p|}-1} \ - \ \frac{e^{|p|}}{(e^{|p|}-1)^2}  \right)dp\ + \ \int_{\mathbb{R}^3}|\varphi|^2\frac{(e^\omega-1)}{4}dp\\ 
 \ge  &  \ \int_{\mathbb{R}^3}|\varphi|^2\frac{(e^\omega-1)}{4}dp.
\end{aligned}
\end{equation}

\bibliographystyle{plain}

 \bibliography{QuantumBoltzmann,QuantumBoltzmann1}

\end{document}